\documentclass[11pt]{article}
\usepackage{fullpage} 

\usepackage[
normalsections, 
normalmargins, 
normallists, 
normalfloats, 
normalindent, 
normaltitle, 
normalleading, 
normallooseness, 
]{savetrees}

\usepackage{amssymb,comment}
\usepackage{amsmath}
\usepackage{tikz}
\usepackage{pgflibrarysnakes}
\usetikzlibrary{snakes}
\usepackage{times}
\usepackage{amstext}
\usepackage{calc}
\usepackage{amsopn}
\usepackage[noend]{algorithmic}
\usepackage[boxed]{algorithm}
\usepackage{eucal}
\usepackage{latexsym}

\usepackage{amsthm}
\usepackage{xspace}

\newtheorem{theorem}{Theorem}
\newtheorem{corollary}[theorem]{Corollary}
\newtheorem{proposition}[theorem]{Proposition}
\newtheorem{lemma}[theorem]{Lemma}

\theoremstyle{definition}
\newtheorem{definition}[theorem]{Definition}
\newtheorem{remark}[theorem]{Remark}
\newtheorem{step}{Step}

\floatname{algorithm}{Algorithm}

\renewcommand{\subset}{\subseteq}

\newcommand{\cliquename}{{\sc{Clique}}\xspace}
\newcommand{\domsetname}{{\sc{Dominating Set}}\xspace}
\newcommand{\mdsname}{{\sc{MDS}}\xspace}
\newcommand{\midsname}{{\sc{MIDS}}\xspace}

\newcommand{\kDS}{$k$-{\sc{Dominating Set}}\xspace}
\newcommand{\kClique}{$k$-{\sc{Clique}}\xspace}

\newcommand{\name}[1]{\textsc{#1}\xspace}

\newcommand{\decnamedefn}[3]{
\begin{tabbing} \name{#1}\\
\emph{Input:} \hspace{1.2cm} \= \parbox[t]{14cm}{#2} \\
\emph{Question:}             \> \parbox[t]{14cm}{#3} \\
\end{tabbing}
}

\newcommand{\parnamedefn}[4]{
\begin{tabbing}
\name{#1}\\
\emph{Input:} \hspace{1.2cm} \= \parbox[t]{14cm}{#2} \\
\emph{Parameter:}            \> \parbox[t]{14cm}{#3} \\
\emph{Question:}             \> \parbox[t]{14cm}{#4} \\
\end{tabbing}
}

\newcommand{\ds}{\name{Dominating Set}}
\newcommand{\rbds}{\name{Red-Blue Dominating Set}}
\newcommand{\crbds}{\name{Colourful Red-Blue Dominating Set}}
\newcommand{\clique}{\name{Clique}}

\newcommand{\wth}{$W[2]$-hard}

\newcommand{\YES}{\textup{\textsf{YES}}\xspace}

\begin{document}

  \date{}

  \author{
    Marek Cygan\thanks{Institute of Informatics,
      University of Warsaw, Poland, \texttt{cygan@mimuw.edu.pl}}
    \and
    Geevarghese Philip\thanks{The Institute of Mathematical Sciences, Chennai, India,
      \texttt{gphilip@imsc.res.in}}
    \and
    Marcin Pilipczuk\thanks{Institute of Informatics,
      University of Warsaw, Poland, \texttt{malcin@mimuw.edu.pl}}
    \and
    Micha\l{} Pilipczuk
      \thanks{Faculty of Mathematics, Computer Science and Mechanics,
      University of Warsaw, Poland, \texttt{michal.pilipczuk@students.mimuw.edu.pl}}
    \and
    Jakub Onufry Wojtaszczyk
    \thanks{
      \texttt{onufry@mimuw.edu.pl} }
  }

  \title{Dominating Set is Fixed Parameter Tractable in Claw-free Graphs\thanks{The authors from the University of Warsaw were partially supported by Polish Ministry of Science grant no. N206 567140 and Foundation for Polish Science. }}

  \maketitle

\begin{abstract}
  We show that the \ds{} problem parameterized by solution size is
  fixed-parameter tractable (FPT) in graphs that do not contain
  the claw (\(K_{1,3}\), the complete bipartite graph on four
  vertices where the two parts have one and three vertices,
  respectively) as an \emph{induced} subgraph. We present an
  algorithm that uses $2^{O(k^2)} n^{O(1)}$ time and polynomial
  space to decide whether a claw-free graph on \(n\) vertices has
  a dominating set of size at most \(k\). Note that this
  parameterization of \ds{} is \wth{} on the set of all graphs,
  and thus is unlikely to have an FPT algorithm for graphs in
  general.

  The most general class of graphs for which an FPT algorithm was
  previously known for this parameterization of \ds{} is the class
  of \(K_{i,j}\)-free graphs, which exclude, for some fixed
  \(i,j\in\mathbb{N}\), the complete bipartite graph \(K_{i,j}\)
  as a \emph{subgraph}. For \(i,j\ge 2\), the class of claw-free
  graphs and any class of \(K_{i,j}\)-free graphs are not
  comparable with respect to set inclusion. We thus \emph{extend}
  the range of graphs over which this parameterization of \ds{} is
  known to be fixed-parameter tractable.

  We also show that, in some sense, it is the presence of the claw
  that makes this parameterization of the \ds{} problem hard. More
  precisely, we show that for any \(t\ge 4\), the \ds{} problem
  parameterized by the solution size is \wth{} in graphs that
  exclude the \(t\)-claw \(K_{1,t}\) as an induced subgraph. Our
  arguments also imply that the related \name{Connected Dominating
    Set} and \name{Dominating Clique} problems are \wth{} in these
  graph classes. 

  Finally, we show that for any \(t\in\mathbb{N}\), the \clique{}
  problem parameterized by solution size, which is \(W[1]\)-hard
  on general graphs, is FPT in \(t\)-claw-free graphs. Our results
  add to the small and growing collection of FPT results for graph
  classes defined by excluded \emph{subgraphs}, rather than by
  excluded \emph{minors}.
  

\end{abstract}

\newcommand{\naglowek}[1]{\noindent {\bf{#1}} }
\newcommand{\claw}[4]{\ensuremath{G[\{#1,#2,#3,#4\}]}}

\section{Introduction} \label{sec:intro}

A {\em dominating set} of a graph $G=(V,E)$ is a set $S\subseteq
V$ of vertices of $G$ such that every vertex in $V\setminus S$ is
adjacent to some vertex in $S$. The \domsetname problem is defined
as:

\decnamedefn{\domsetname}{A graph $G=(V,E)$ and a non-negative
  integer $k$.}{Does $G$ have a dominating set with \emph{at most}
  $k$ vertices?}

A {\em clique} in a graph $G=(V,E)$ is a set $C\subseteq V$ of
vertices of $G$ such that there is an edge in \(G\) between any
two vertices in \(C\). The \cliquename problem is defined as:

\decnamedefn{\cliquename}{A graph $G=(V,E)$ and a non-negative
  integer $k$.}{Does $G$ contain a clique with \emph{at least} $k$
  vertices?}

The \domsetname and \cliquename problems are both classical
NP-hard problems, belonging to Karp's original
list~\cite{Karp1972} of 21 NP-complete problems. These problems
were later shown to be NP-hard even in very restricted graph
classes, such as the class of planar graphs with maximum
degree~$3$~\cite{GareyJohnson1979} for \domsetname, and the class
of \(t\)-interval graphs for any \(t\ge 3\) for
\cliquename~\cite{ButmanHermelinLewensteinRawitz2010}. Hence,
unless \(\text{P}=\text{NP}\), there is no polynomial-time
algorithm that solves these problems even in such restricted graph
classes.

Parameterized
algorithms~\cite{downey-fellows,FlumGroheBook,NiedermeierBook}
constitute one approach towards solving NP-hard problems in
``feasible'' time. Each parameterized problem comes with an
associated {\em parameter}, which is usually a non-negative
integer, and the goal is to find algorithms that solve the problem
in polynomial time {\em when the parameter is fixed}, where the
degree of the polynomial is independent of the parameter. More
precisely, if $k$ is the parameter and $n$ the size of the input,
then the goal is to obtain an algorithm that solves the problem in
time \(f(k)\cdot n^{c}\) where $f$ is some computable function and
$c$ is a constant independent of $k$. Such an algorithm is called
a fixed-parameter-tractable (FPT) algorithm, and the class of all
parameterized problems that have FPT algorithms is called FPT; a
parameterized problem that has a fixed-parameter-tractable
algorithm is said to be (in) FPT.

Together with this revised notion of tractability, parameterized
complexity theory offers a corresponding notion of intractability
as well, captured by the concept of {\em $W$-hardness}. In brief,
the theory defines a hierarchy of complexity classes $\text{FPT} \subset
\text{W[1]} \subset \text{W[2]} \cdots \subset \text{XP}$, where each
inclusion is believed to be strict --- on the basis of evidence
similar in spirit to the evidence for believing that
\(\text{P}\ne\text{NP}\) ---
and XP is the class of all parameterized problems that can be
solved in $O(n^{f(k)})$ time where $n$ is the input size, $k$ the
parameter, and $f$ is some computable
function~\cite{downey-fellows,FlumGroheBook}.

A natural parameter for both \domsetname and \cliquename is $k$,
the size of the solution being sought. Natural parameterized
versions of these problems are thus the \kDS and \kClique
problems, defined as follows:
 
\parnamedefn{\kDS}{A graph $G=(V,E)$, and a non-negative integer
  $k$.}{$k$}{Does $G$ have a dominating set with \emph{at most} $k$ vertices?}

\parnamedefn{\kClique}{A graph $G=(V,E)$ and a non-negative
  integer $k$.}{$k$}{Does $G$ contain a clique with \emph{at least} $k$
  vertices?}

It turns out that both the \domsetname and \cliquename problems,
with these parameterizations, are still hard to solve. More
precisely, \kDS is the canonical W[2]-hard problem, and \kClique
is the canonical W[1]-hard problem~\cite{downey-fellows}. Thus
there are no FPT algorithms that solve these problems unless
\(\text{FPT}=\text{W[2]}\) and \(\text{FPT}=\text{W[1]}\),
respectively, which are both considered unlikely.

These problems do become easier in the parameterized sense when
the input is restricted to certain classes of graphs. Thus, the
\kDS problem has FPT algorithms in planar
graphs~\cite{domset-planar}, graphs of bounded
genus~\cite{domset-genus}, nowhere-dense classes of
graphs\cite{DawarKreutzer2009}, $K_{h}$-topological-minor-free
graphs and graphs of bounded degeneracy~\cite{domset-degenerate},
and in \(K_{i,j}\)-free graphs~\cite{domset-philip}. It is easily
observed that \kClique has an FPT algorithm in \emph{any} class of
graphs characterized by a finite set of excluded minors or
excluded subgraphs; this includes all the classes mentioned above
and many more.

A number of powerful tools that yield FPT algorithms are based on
encoding problems in terms of formulas in different logics.  Much
effort has gone into understanding the parameterized complexity of
evaluating logic formulas on {\em{sparse}} graphs, where the
length of the formula is the parameter.  A stellar example is the
celebrated theorem by Courcelle~\cite{courcelle} which states that
\emph{any} problem that can be expressed in Monadic Second-Order
Logic has FPT algorithms when restricted to graphs of bounded
treewidth.  Similarly, a sequence of papers gives FPT algorithms
for problems expressible in First-Order Logic on graph classes of
bounded degree~\cite{seese:bounded-degree}, bounded local
treewidth~\cite{frick-grohe:bounded-local-tw}, excluding a
minor~\cite{flum-grohe:exclude-minor}, locally excluding a
minor~\cite{dawar:FO-local-minor}, and classes of bounded
expansion~\cite{kral:FO-bounded-exp}. Note that the existence of a
clique (resp. dominating set) of size $k$ can be expressed as a
first order formula of length $O(k^2)$ (resp.~$O(k)$), and so both
\kClique and \kDS are FPT on the aforementioned classes of sparse
graphs.

The \emph{claw} is the complete bipartite graph \(K_{1,3}\), which
has a single vertex in one part and three in the other part of the
bipartition. \emph{Claw-free} graphs are undirected graphs which
exclude the claw as an induced subgraph. Equivalently, an
undirected graph is claw-free if it does not contain a vertex with
three pairwise nonadjacent neighbours. Claw-free graphs are a
generalization of \emph{line} graphs, and they have been
extensively studied from the graph-theoretic and algorithmic
points of view --- see the survey by Faudree et
al.~\cite{FaudreeFlandrinRyjacek1997} for a summary of the main
results. More recently, Chudnovsky and
Seymour~\cite{chudnovsky:clawfree1,chudnovsky:clawfree2,chudnovsky:clawfree3,chudnovsky:clawfree4,chudnovsky:clawfree5,chudnovsky:clawfree6}
developed a structure theory for this class of graphs, analogous
to the celebrated graph structure theorem for minor-closed graph
families proved earlier by Robertson and
Seymour~\cite{RobertsonSeymour2003}. While some problems which are
NP-hard in general graphs (e.g.: Maximum Independent Set) become
solvable in polynomial time in claw-free
graphs~\cite{FaudreeFlandrinRyjacek1997}, it turns out that both
\domsetname~\cite{HedetniemiLaskar1988} and
\cliquename~\cite{FaudreeFlandrinRyjacek1997} are NP-hard on
claw-free graphs.

\paragraph{Our Results.}

$K_{i,j}$ denotes the complete bipartite graph on $i+j$ vertices
where one piece of the partition has $i$ vertices and the other
part has $j$. A graph is said to be {\em $K_{i,j}$-free} if it
does not contain $K_{i,j}$ as a (not necessarily induced)
subgraph. To the best of our knowledge, \(K_{i,j}\)-free graphs
are the most general graph classes currently
known~\cite{domset-philip} to have an FPT algorithm for the \kDS
problem. Observe that in the interesting case when \(i,j\ge 2\),
the class of claw-free graphs is not comparable --- with respect
to set inclusion --- with any class of \(K_{i,j}\)-free graphs: a
\(K_{i,j}\)-free graph can contain a claw, and a claw-free graph
can contain a \(K_{i,j}\) as a subgraph. In the main result of
this paper, we show that \kDS is FPT in claw-free graphs:

\begin{theorem}\label{thm:ds-claw-free-fpt}
  The \kDS problem can be solved in \(2^{O(k^2)} n^{O(1)}\) time
  and using \(n^{O(1)}\) space.
\end{theorem}

 We thus
\emph{extend} the range of graphs in which \kDS is FPT, to beyond
classes that can be described as \(K_{i,j}\)-free.

For \(t\in\mathbb{N}\), the \(t\)-claw is the graph
\(K_{1,t}\). Given that \kDS is FPT in claw-free graphs, one
natural question to ask is whether the problem remains FPT in
graphs that exclude larger claws as induced subgraphs. We show
that this is indeed \emph{not} the case; the presence of the
(\(3\)-)claw is what makes the problem W[2]-hard, in the following
sense:

\begin{theorem}\label{thm:ds-t-claw-free-whard}
  For any \(t\ge 4\), the \kDS problem is W[2]-hard in graphs
  which exclude the \(t\)-claw as an induced subgraph.
\end{theorem}

Our third and final result is to show that --- as might perhaps be
expected --- excluding a claw of any size renders the \kClique
problem FPT:

\begin{theorem}\label{thm:clique-fpt}
  For any \(t\ge 3\), the \kClique problem is FPT in graphs which
  exclude the \(t\)-claw as an induced subgraph.
\end{theorem}

\paragraph{Recent Developments.}
Building on the structural characterization for claw-free graphs
developed recently by Chudnovsky and Seymour, Hermelin et
al.~\cite{HermelinMnichLeeuwenWoeginger2011} have developed a
faster FPT algorithm for the \kDS{} problem on claw-free graphs
which runs in \(9^{k}n^{O(1)}\) time. They have also shown that
the problem has a polynomial kernel on \(O(k^{4})\) vertices on
claw-free graphs.

\paragraph{Organization of the rest of the paper.}

We describe the basic notation used in this paper in the next paragraph.
We prove Theorem~\ref{thm:ds-claw-free-fpt} in
Section~\ref{s:alg}, Theorem~\ref{thm:ds-t-claw-free-whard} in
Section~\ref{s:hardness}, and Theorem~\ref{thm:clique-fpt} in
Section~\ref{s:clique}. We conclude and list some open problems in
Section~\ref{s:conclusions}.

\paragraph{Notation.}

In this paper all graphs are undirected. In Section \ref{s:alg} we silently assume
that the input instance is a claw-free graph $G=(V,E)$ together with a parameter $k$.
For any vertex set $X \subset V$, by $G[X]$ we denote the subgraph induced by $X$. For any $v \in V$ by $N(v)$ we denote
the set of neighbours of $v$, and by $N[v] = \{v\} \cup N(v)$ the closed neighbourhood of $v$. We extend this notation
to sets of vertices $X \subset V$: $N[X] = \bigcup_{v \in X} N[v]$, $N(X) = N[X] \setminus X$.

In our proofs we often look at groups of four vertices and deduce (non)existence of some edges by the fact that
these four vertices do not induce a claw ($K_{1,3}$). By saying that quadruple $\claw{v}{x}{y}{z}$ risk a claw
we mean that we use the fact that we cannot have at once $vx,vy,vz \in E$ and $xy,yz,xz \notin E$.

By \mdsname{} we mean minimum dominating set. We sometimes look at dominating sets that
are also independent sets (in other words, inclusion-maximal independent sets).
By \midsname{} we mean minimum independent dominating set. It is well-known that
in claw-free graphs the sizes of \mdsname{} and \midsname{} coincide; we prove
this result in a bit stronger form in Section \ref{ss:indset}.

For vertex sets \(A,B\subseteq V\) of a graph \(G=(V,E)\), we say
that \(A\) is a dominating set of \(B\) if every vertex in \(B \setminus A\)
has at least one neighbour in \(A\).

\section{Finding minimum dominating set in claw-free graphs}\label{s:alg}

In this section we prove Theorem \ref{thm:ds-claw-free-fpt}, i.e.,
we present an algorithm that checks whether a given claw-free graph $G=(V,E)$
has a dominating set of size at most $k$. The algorithm runs in $2^{O(k^2)} n^{O(1)}$ time
and uses polynomial space.

The general idea of the algorithm is as follows. In Section \ref{ss:indset}
we find (in polynomial time) the largest independent set $I$ in $G$.
It turns out to be of size $O(k)$.
We branch --- if a solution intersects $I$, we guess the intersection and
reduce $k$. From now on we assume that the solution is disjoint with $I$.

In Section \ref{ss:packs} we learn that the set $I$ introduces a
structure of $O(k^2)$ {\em packs} on the remaining vertices of $G$. In
Section \ref{ss:structure} we branch again, guessing the layout of the
solution within the packs. It turns out that at most one vertex of the
solution can lie within each pack.

In Section \ref{ss:overview} we start eliminating vertices. We introduce
a notation to mark vertices that are sure to be dominated, no matter how we
choose our solution, and vertices which are sure not to be included in any
solution. We show several simple rules to move vertices to these groups.
Then, in Section \ref{ss:$1$-packs}, we perform a thorough analysis of a
more difficult type of packs --- the $1$-packs --- and significantly prune
the vertices to consider in them.

In a perfect world, all the pruning would leave us only with a single
possible solution (or at most $f(k)$ possible solutions, which we could
directly check). This is not, however, the case --- we can be left with a
large number of potential solutions. The trick we use is to notice our
choices of vertices included in the solution from each pack are close to
independent, which will allow us to use dynamic programming approach to solve
the problem, formalized as an auxiliary CSP introduced in Section
\ref{ss:dp}. We will need to simplify the constraints before this works,
and the simplification occurs in Section \ref{ss:degred}.

The algorithm is rather complex, and involves a number of technical details.
Thus, we included a more detailed summary of what actually happens in
Section \ref{ss:summary}.
The best way to get an idea what really happens would probably be to read
and understand all the definitions and statements of the algorithm in
Sections \ref{ss:indset}--\ref{ss:degred}, then go over the summary in
Section \ref{ss:summary}, and finally come back and fill in all the proofs.

\subsection{Maximum independent set}\label{ss:indset}

We start with a folklore fact showing that the sizes of a minimum dominating set (\mdsname{}) and a minimum independent dominating set (\midsname{})
coincide in claw-free graphs.

\begin{lemma}\label{lem:nv-nw-clique}
  Let $G=(V,E)$ be a claw-free graph and let $vw \in E$.
  Then $G[(N[w] \setminus N[v]) \cup w]$ is a clique.
\end{lemma}

\begin{proof}
  Assume that there are some two vertices $x,y \in w \cup (N[w] \setminus N[v])$
  with no edge between them. The vertex $w$ is connected to all the other vertices,
  as they are in $N[w]$, so $x,y \neq w$.
  We have $wv \in E$ (from our assumptions) and $wx, wy \in E$ (as $x,y \in N[w] \setminus w$).
  However $xy \not\in E$ from their definition, and $vx, vy \not\in E$ as $x,y \in N[w] \setminus
  N[v]$. Thus $G[\{w,v,x,y\}]$ is a claw, contradicting the assumption on $G$.
\end{proof}

\begin{proposition}\label{prop:mds-vs-mids}
  Let $D$ be any dominating set in a claw-free graph $G$ and let $I_D \subset D$ be any independent
  set of vertices in $D$. Then there exists an independent dominating set $D'$ such that
  $|D'| \leq |D|$ and $I_D \subset D'$.
\end{proposition}

\begin{proof}
  Let $D'$ be an inclusion-minimal dominating set of $G$
  satisfying the following three properties: (a) $|D'| \leq |D|$,
  (b) $I_D \subset D'$ and (c) $G[D']$ has the smallest possible
  number of edges.  Since $D$ satisfies the first two properties,
  such a $D'$ is guaranteed to exist.  Suppose $D'$ is not an
  independent dominating set, i.e., there exists $vw \in E; v,w\in
  D'$. Since $I_D$ is an independent set, both $v$ and $w$ cannot
  be at once in $I_D$, so let us assume that $w \notin I_D$.  Let
  $X$ be the set of vertices in $G$ which are {\it not} dominated
  by $D' \setminus \{w\}$.  From the minimality of $D'$, the set
  $X$ is nonempty.  Since $X \subset N[w] \setminus N[v]$, by
  Lemma \ref{lem:nv-nw-clique} $G[X]$ is a clique.  Let $D'' = D'
  \setminus \{w\} \cup \{x\}$, where $x$ is an arbitrary vertex in
  $X$. Then $|D''| = |D'|$, $I_D \subset D''$ as $w \notin I_D$,
  $D''$ is a dominating set of $G$. 
  Observe that \(x\) has degree zero in \(G[D'']\), while \(w\)
  has degree at least one in \(G[D']\). This implies that
  \(G[D'']\) has fewer edges than \(G[D']\), a contradiction.
\end{proof}

Therefore, it is sufficient to look for an independent dominating set of size at most $k$.
The following lemma shows that this assumption can simplify our algorithm ---
if we decide to include some vertex $v$
in the solution, we can simply delete $N[v]$ from the graph and decrease $k$ by one.

\begin{lemma} \label{lem:reduce-problem}
  Let $G = (V,E)$ be a claw-free graph and let $v \in V$. There exists a \midsname of size at most $k$ containing $v$
  if and only if there exists a \midsname of size at most $k-1$ in $G[V \setminus N[v]]$.
\end{lemma}

\begin{proof}
  Suppose we have a \midsname $D$ in $G$ of size $k$ and containing $v$. The set $D \setminus \{v\}$
  is disjoint from $N[v]$ (as $D$ is an independent set), and dominates $V \setminus N[v]$
  (as $D$ is a dominating set), and thus is a \midsname of size $k-1$ in $G[V\setminus N[v]]$.

  Conversely, consider any \midsname $D'$ of size $k-1$ in $G[V\setminus N[v]]$.
  Then $D' \cup \{v\}$ is independent in $G$ (as $D'$ lies outside $N[v]$), and dominates
  $V$ (as $D'$ dominates $V \setminus N[v]$ and $v$ dominates $N[v]$), and thus is a \midsname
  of size $k$ in $G$.
\end{proof}

We now start describing our algorithm. The algorithm is presented as a sequence of steps.
\begin{step}\label{step:1}
Find a largest independent set $I$ in $G$.
This can be done in polynomial time in claw-free graphs \cite{sbihi:indset,minty:indset}.
\end{step}
If $I$ is too small or too large, we may quit immediately.
\begin{step}
If $|I| \leq k$, return YES, since $I$ is a dominating set as well; in any graph, any maximal independent set is also a dominating set.
If $|I| > 2k$, return NO.
\end{step}
The following lemma justifies the above step:
\begin{lemma}
  Let $G=(V,E)$ be a claw-free graph, and let $I$ be a largest independent set in $G$.
  Then any dominating set in $G$ contains at least $|I|/2$ vertices.
\end{lemma}
\begin{proof}
  Assume we have a dominating set $D$ with $|D| < |I|\slash 2$. In particular, $D$ has to
  dominate $I$, and by the pigeonhole principle, there exists a vertex $v \in D$ that dominates
  at least three vertices $x,y,z$ from $I$. Notice that a vertex from $I$ does not dominate
  any other vertex from $I$, as $I$ is independent, so $v \not\in I$, and in particular
  $v \not \in \{x,y,z\}$. But now $G[\{v,x,y,z\}]$ is a claw --- we have $vx, vy, vz \in E$, as
  $v$ dominates $\{x,y,z\}$, but $xy,yz,zx \not\in E$ as $x,y,z \in I$ and $I$ is independent.
  The contradiction ends the proof.
\end{proof}

\begin{step}\label{step-Idisjoint}
  Now the algorithm branches into the following two cases:
\begin{enumerate}
  \item There exists an \midsname{} with a nonempty intersection with $I$.
  \item Every \midsname{} in $G$ is disjoint with $I$.
\end{enumerate}
In the first case, the algorithm simply guesses any single vertex from the intersection,
deletes its closed neighbourhood, decreases $k$ by one and goes back to Step \ref{step:1}.
\end{step}
Note that since we are aiming for the time complexity $2^{O(k^2)} n^{O(1)}$,
the branching in the first case fits into the time bound.

From now on, in the algorithm we assume that every \midsname{} in $G$ is disjoint with $I$.
Note that the algorithm does not explicitly check whether this condition is true --- instead, if at any subsequent point any conclusion from this assumption
appears to be wrong, the algorithm merely terminates that branch of the computation.

\subsection{Packs}\label{ss:packs}

Note that for each $v \in V \setminus I$ the vertex $v$ knows at least one vertex from $I$
(since $I$ is maximum hence maximal) and knows at most two vertices from $I$ (since
otherwise they form a claw, as $I$ is an independent set). Thus we may partition $V \setminus I$
into the following parts.
\begin{definition}
  For each $a,b \in I, a \not = b$ we denote $V_{a,b} = \{v \in V \setminus I: N(v) \cap I = \{a,b\}\}$
  and $V_a = \{v \in V \setminus I: N(v) \cap I = \{a\}\}$. 
  The sets $V_{a,b}$ and $V_a$ are called {\em{packs}}. The sets
  $V_a$ are called {\em{$1$-packs}} and the sets $V_{a,b}$ are called {\em{$2$-packs}}.
  For a pack $V_{a,b}$ ($V_a$) the vertices $a$ and $b$ (the vertex $a$) are called
  {\em{legs}} (the {\em{leg}}) of the pack.
\end{definition}

See Figure~\ref{fig:packs} for an illustration.
  \begin{figure}[t]
    \centering
    \includegraphics[clip,scale=0.4]{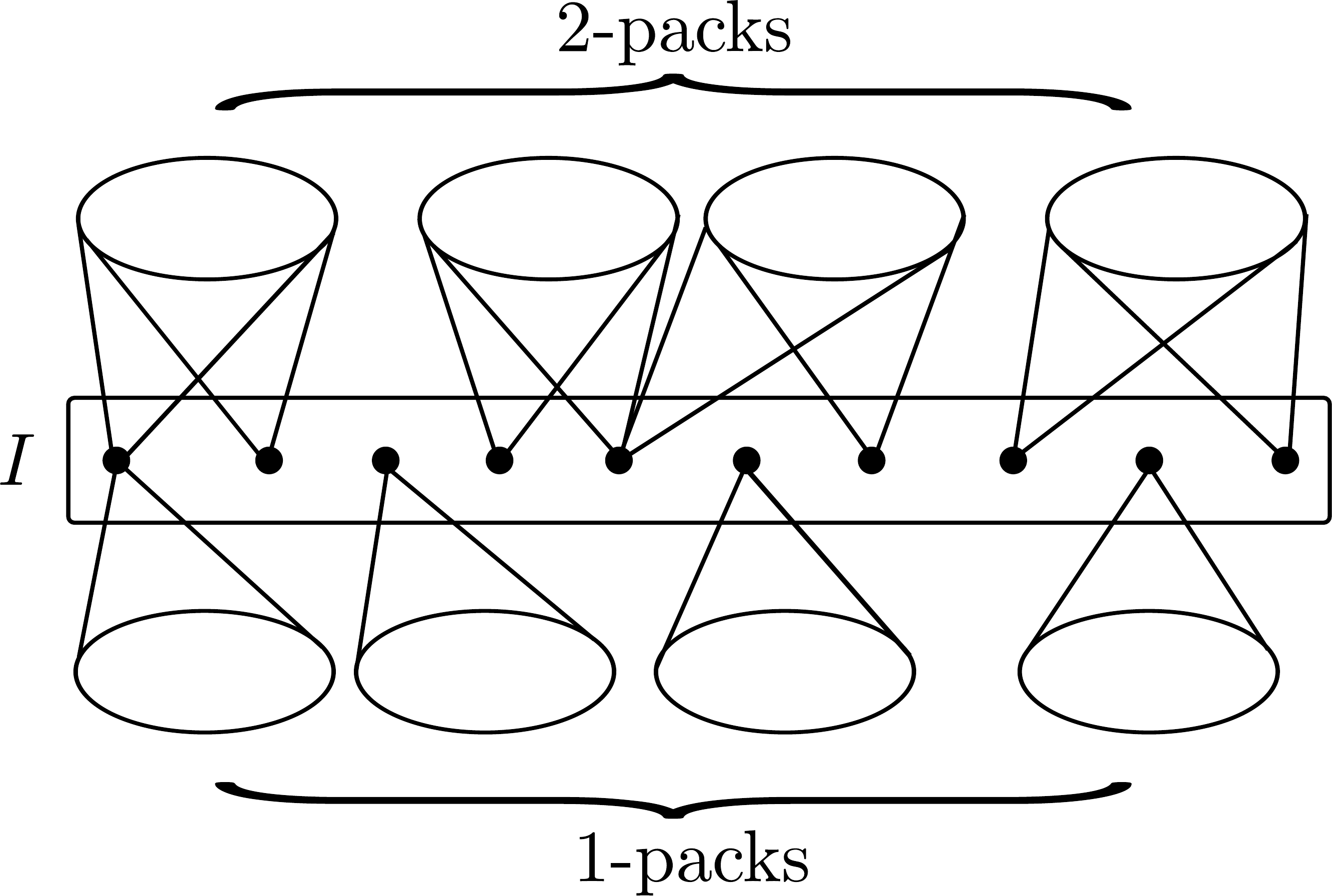}
    \caption{\label{fig:packs}A schematic diagram showing the
      two kinds of packs. \(I\) is a maximum independent
      set. Edges with end-points in different packs may be
      present in the graph; these are not shown in this diagram.}
  \end{figure}

\begin{lemma} \label{lem:balloon-clique}
  For any $1$-pack $V_a$, $G[V_a]$ is a clique.
\end{lemma}
\begin{proof}
  Assume we have $x,y \in V_a$, with $xy \not\in E$. Consider $(I  \setminus \{a\}) \cup \{x,y\}$.
  This is an independent set --- $I \setminus \{a\}$ is independent, $x$ and $y$ have no edges
  to $I \setminus \{a\}$ from the definition of $V_a$, and there is no edge between them. But this
  set is larger than $I$, contradicting the definition of $I$.
\end{proof}

\begin{lemma}\label{lem:zeppelin-knows-locally}
  If there is an edge between a $2$-pack $V_{a,b}$ and a distinct pack $X$, then $X$ and $V_{a,b}$
  have a common leg, i.e., $X = V_a$ or $X = V_b$ or $X = V_{a,c}$ or $X = V_{b,c}$ for some
  $c \in I$.
\end{lemma}
\begin{proof}
  Suppose not. Let $vw$ be the edge between $V_{a,b}$ and $X$, with $v \in V_{a,b}$ and $w \in X$. We know
  that $wa, wb \not\in E$, as $X$ has no common leg with $V_{a,b}$. Moreover, $ab \not\in E$ as
  they both belong to the independent set $I$, and $va, vb, vw \in E$ (first two from the definition
  of $V_{a,b}$, the third from the assumptions). Thus $G[\{v,w,a,b\}]$ is a claw, a contradiction.
\end{proof}

\newcommand{\uppacks}{\ensuremath{\mathcal{B}}}

\subsection{Solution structure}\label{ss:structure}

We now analyze how a \midsname{} can be placed with respect to $1$-packs and $2$-packs.

\begin{lemma}\label{lem:two-dom-nv}
  Let $v \in V$ and $w_1, w_2 \in N(v)$, $w_1w_2 \notin E$. Then
  $N[v] \subset N[w_1] \cup N[w_2]$, i.e., $w_1$ and $w_2$ dominate everything that
  $v$ dominates.
\end{lemma}
\begin{proof}
  Assume there is a vertex $w_3 \in N[v] \setminus (N[w_1] \cup N[w_2])$.
  We have $w_1 \in N(v)$, so $v \in N[w_1]$ and so $w_3 \neq v$. But now
  $w_1w_2, w_2w_3, w_3w_1 \not\in E$, the first from the assumptions, the other two
  by the definition of $w_3$. On the other hand, $vw_1, vw_2, vw_3 \in E$, thus
  $G[\{v,w_1,w_2,w_3\}]$ is a claw, a contradiction.
\end{proof}

\begin{lemma} \label{lem-square-domination}
  Let $v_1, v_2 \in V$, $v_1v_2 \notin E$. Let $w_1, w_2 \in N(v_1) \cap N(v_2)$, $w_1w_2 \notin E$.
  Then $N[v_1] \cup N[v_2] = N[w_1] \cup N[w_2]$, i.e., $v_1$ and $v_2$ dominate together exactly the
  same vertex set as $w_1$ and $w_2$.
\end{lemma}
\begin{proof}
  Using Lemma \ref{lem:two-dom-nv} four times we obtain that
  $N[v_1], N[v_2] \subset N[w_1] \cup N[w_2]$ and
  $N[w_1], N[w_2] \subset N[v_1] \cup N[v_2]$.
\end{proof}

\begin{lemma} \label{lem-packs-size-one}
  Assume there exists a \midsname{} $D$ and a pack $X$, such that $|D \cap X| > 1$.
  Then there exists a \midsname{} $D'$ that is not disjoint with $I$.
\end{lemma}
\begin{proof}
  By Lemma \ref{lem:balloon-clique}, all $1$-packs are cliques, so we cannot have two vertices
  from the independent set $D$ in $X$. Thus $X = V_{a,b}$ for some $a,b \in I$. Let $v,w \in
  D \cap X$.
  Now $vw \not\in E$ as $D$ is independent, $ab \not\in E$ as $I$ is independent, and
  $a,b \in N(v) \cap N(w)$ by the definition of $V_{a,b}$. Thus the assumptions of Lemma
  \ref{lem-square-domination} are satisfied, and so $N[a] \cup N[b] = N[v] \cup N[w]$. Thus
  $D' = (D \setminus \{v,w\}) \cup \{a,b\}$ is a dominating set.

  Now we apply Proposition \ref{prop:mds-vs-mids}. We have a dominating set $D'$ with
  $|D'| = |D|$, and
  an independent set $\{a,b\} \subset D'$. Proposition \ref{prop:mds-vs-mids} guarantees
  the existence of an independent dominating set $D''$ with $|D''| \leq |D'|$ and
  $\{a,b\} \subset D''$. As $D$ was a \midsname, however, we have $|D''| \geq |D|$, and
  thus $|D''| = |D|$ --- thus $D''$ is also a \midsname, and is not disjoint with $I$.
\end{proof}

Recall from the discussion at the end of Section~\ref{ss:indset}
that we may assume, without loss of generality, that \emph{every}
\midsname{} in the graph \(G\) is disjoint with the set \(I\).  It
follows from Lemma~\ref{lem-packs-size-one} that every pack
contains at most one vertex from the solution. We limit
ourselves to this case in the remaining part of the algorithm.

\begin{definition}
  We say that a \midsname $D$ is {\em{compatible}} with a set $\uppacks$ of packs,
  if $D$ contains exactly one vertex in each pack in $\uppacks$, and no vertices
  in the packs not in $\uppacks$.
\end{definition}

\begin{step} \label{step-guess-uppacks}
  The algorithm now guesses a set $\uppacks$ of at most $k$ packs. From now on,
  the algorithm looks for a \midsname{} compatible with $\uppacks$.
\end{step}

As the number of packs is at most $2k + \binom{2k}{2} = O(k^2)$, we have $2^{O(k \log k)}$
possible guesses.

Some guesses are clearly invalid.
\begin{step}\label{step-stupid-uppacks}
  The algorithm discards guesses in which:
\begin{enumerate}
  \item there exists a vertex $a \in I$ that cannot be dominated, i.e., no pack with leg $a$
    is chosen to be in $\uppacks$;
  \item or there exists a vertex $a \in I$, such that at least three packs with leg $a$
    are chosen to be in $\uppacks$ (we cannot find three independent vertices in $N(a)$,
    as they would make a claw with the center in vertex $a$).
\end{enumerate}
\end{step}
To sum up, for each $a \in I$ there exist one or two packs in $\uppacks$ that have a leg $a$.

\newcommand{\Vactive}{\ensuremath{V^{\mathtt{Active}}}}
\newcommand{\Vdone}{\ensuremath{V^\mathtt{Done}}}
\newcommand{\Vpassive}{\ensuremath{V^\mathtt{Passive}}}

\subsection{Algorithm structure}\label{ss:overview}

From now on, the algorithm maintains the partition of the vertex set $V$ into
three parts:
\begin{enumerate}
  \item $\Vactive$, vertices that can be chosen into the constructed \midsname{},
    and we need to dominate them;
  \item $\Vpassive$, vertices that cannot be chosen into the constructed \midsname{},
    but we need to dominate them;
  \item $\Vdone$, vertices that cannot be chosen into the constructed \midsname{},
    and we somehow have ensured that they would be dominated, i.e.,
    we do not need to care about them.
\end{enumerate}
As we show later in this section (See Lemma~\ref{lem:end-24}), it
turns out that it is sufficient to look for a solution which is
``mostly'' --- and not necessarily totally --- an independent
set. More precisely, it is sufficient to find a ``dominating
candidate'' which is also a dominating set:
\begin{definition}\label{def-dominating-candidate}
  A set $D \subset \Vactive$ is called a {\em{dominating candidate}} if it satisfies the following properties:
  \begin{enumerate}
    \item $|D| = |\uppacks|$ and $D$ consists of exactly one active vertex from each pack in $\uppacks$;
    \item if $X,Y \in \uppacks$ and $X$ and $Y$ share a leg, then the two vertices in $D \cap (X \cup Y)$ are nonadjacent.
  \end{enumerate}
  We say that the partition $(\Vactive, \Vpassive, \Vdone)$ is {\em{safe}} if every dominating candidate
  dominates $\Vactive \cup \Vdone$.
\end{definition}

Let \(D\) be a dominating candidate, let $X,Y \in \uppacks$, and
let \(x,y\) be the vertices in \(X,Y\) respectively which are
present in \(D\). Further, let \(xy\) be an edge in the graph. If
\(X\) is a $2$-pack, then by
Lemma~\ref{lem:zeppelin-knows-locally} the packs \(X\) and \(Y\)
share a leg. The second condition in the definition of a
dominating candidate then implies that there is no edge between
\(x\) and \(y\), a contradiction. Thus both \(X\) and \(Y\) are
$1$-packs. Therefore, while the subgraph induced by a dominating
candidate may contain edges, any such edge is between vertices
which belong to distinct $1$-packs. As we see in
Lemma~\ref{lem:end-24}, this relaxation in the independence
requirement for vertices drawn from $1$-packs helps in the
justification of Step~\ref{step:common-leg} below.

At the end of this section we obtain a state where the partition $(\Vactive, \Vpassive, \Vdone)$ is safe.

Initially, $\Vactive$ consists of vertices in packs in $\uppacks$,
$\Vdone = I$ and $\Vpassive = V \setminus (I \cup \Vactive)$
(we do not need to care about $I$, since we have discarded choices of $\uppacks$
that do not dominate whole $I$). Thus, every dominating candidate dominates $\Vdone$,
but not necessarily $\Vactive$. During the whole algorithm we shall keep the invariant that
all active vertices are in $\bigcup \uppacks$ and all passive vertices are in
$V\setminus I \setminus \bigcup \uppacks$.

In the following set of steps we assign some vertices to $\Vdone$ (keeping the invariant
that every dominating candidate dominates $\Vdone$) and assure that every dominating candidate
dominates $\Vactive$. This is formally justified in
Lemma~\ref{lem:end-24}. 

\begin{lemma} \label{lem-useless-baloon}
  Let $v \in V_a \in \uppacks$ and assume that $N[v] \subset N[a]$, i.e., $v$ knows only $a$ and
  vertices from packs that have leg $a$. Then, if there exists an \midsname{} $D$ compatible with
  $\uppacks$ containing $v$, then there exists an \midsname{} $D'$ of cardinality not larger than $D$
  that is not disjoint with $I$.
\end{lemma}
\begin{proof}
  We proceed as in the proof of Lemma \ref{lem-packs-size-one}. Consider
  the set $D' = (D\cup\{a\}) \setminus \{v\}$.
  This is a dominating set, as $N[v] \subset N[a]$. As $\{a\}$ is an independent set,
  by Proposition \ref{prop:mds-vs-mids} we can obtain a \midsname $D''$ not larger than
  $D'$ (and thus not larger than $D$), which contains $a$. That ends the proof.
\end{proof}
This Lemma will be used in the justification of the following step:
\begin{step} \label{step-useless-baloon} \label{step:up-T0}
  For each $v \in V_a \in \uppacks$ such that $N[v] \subset N[a]$, move $v$ to $\Vdone$.
\end{step}

We will now focus on packs that are {\em{alone}} in $\uppacks$:
\begin{definition}
  A pack $X \in \uppacks$ is called {\em{alone}} if for any leg $a$ of $X$
  no other pack $Y \in \uppacks$ has this leg.
\end{definition}

\begin{step} \label{step-lone-$2$-pack}
  Let $V_{a,b} \in \uppacks$ be an alone $2$-pack in $\uppacks$.
  For each vertex $v \in V_{a,b}$, if $V_{a,b}$
  is not dominated by $v$, move $v$ to $\Vdone$.
\end{step}

Finally we remove several vertices from $\Vpassive$:
\begin{step}\label{step:common-leg}
  Let $X, Y \in \uppacks$ be two packs that share a common leg
  $a \in I$.  For each pack $Z \notin \uppacks$ that has the
  leg $a$, move all vertices in $Z$ to $\Vdone$.
\end{step}

We justify all the above steps and formally prove that the current partition
$(\Vactive, \Vpassive, \Vdone)$ is safe in the following lemma:
\begin{lemma}\label{lem:end-24}
  Assume we have finished all steps up to Step \ref{step:common-leg}.
  \begin{enumerate}
    \item $\Vactive \subset \bigcup \uppacks$ and $\Vpassive \subset V \setminus (I \cup \bigcup\uppacks)$.
    \item The partition $(\Vactive, \Vpassive, \Vdone)$ is safe, i.e., every dominating candidate $D$ dominates $\Vactive \cup \Vdone$.
    \item If there exists a \midsname $D$ compatible with $\uppacks$, then there exists a dominating candidate that is a dominating set in $G$.
  \end{enumerate}
\end{lemma}
\begin{proof}
  The first claim is obvious, as in all above steps we only moved vertices from $\Vactive$ or $\Vpassive$ to $\Vdone$.

  First note that in all of the above steps, we only transferred vertices into $\Vdone$,
  in particular if a vertex is in $\Vactive$, it had to be in $\Vactive$ at the start,
  and so is in one of the packs from $\uppacks$.
  Consider any dominating candidate $D$, and any vertex $v \in \Vactive$. Let $X$ be the
  pack containing $v$. Observe that $X \in \uppacks$, as $v\in \Vactive$ and all the vertices in
  packs not in $\uppacks$ were outside $\Vactive$ from the beginning

  Let $X = V_{a}$, i.e., let $X$ be a $1$-pack. This means
  $D$ contains a vertex $x \in V_a$, and --- as $V_a$ is a clique by Lemma \ref{lem:balloon-clique}
  --- $v$ is dominated by $x$.

  Now consider the case when $X = V_{a,b}$, i.e., $X$ is a $2$-pack. As before, $X \in \uppacks$, and
  let $x$ be the vertex in $D \cap X$. As $D$ is a dominating candidate, $x \in \Vactive$.
  If $X$ is alone, then $x$ dominates
  $V_{a,b}$ --- otherwise it would be removed from $\Vactive$ in Step \ref{step-lone-$2$-pack}
  --- and thus in particular $x$ dominates $v$.
  If $X$ is not alone, then we have another pack $Y \in \uppacks$ that shares a leg, say $a$,
  with $X$, and a vertex $y \in Y \cap D$. As $D$ is a dominating candidate, $xy \not\in E$, and
  both $x$ and $y$ are adjacent to $a$.
  Thus, by Lemma \ref{lem:two-dom-nv}, $\{x,y\}$ dominates $N[a]$, and --- in particular --- $v$.

  The above proves that $\Vactive$ is indeed dominated by $D$. Now consider a vertex $v \in \Vdone$.
  If $v \in I$, then $v$ is dominated by every dominating candidate, as we disregarded choices
  of $\uppacks$ that do not guarantee this in Step \ref{step-stupid-uppacks}. We thus have
  to consider vertices moved to $\Vdone$ in Steps \ref{step-useless-baloon}--\ref{step:common-leg}.

  If $v$ was moved to $\Vdone$ in Step \ref{step-useless-baloon}, then $v \in V_a$, with $V_a \in
  \uppacks$. Thus there exists a vertex $x \in V_a \cap D$, and this vertex dominates $v$ as
  $V_a$ is a clique by Lemma \ref{lem:balloon-clique}.

  If $v$ was moved to $\Vdone$ in Step \ref{step-lone-$2$-pack}, then $v \in V_{a,b}$,
  $V_{a,b} \in \uppacks$ and $V_{a,b}$ was alone. Again, we have a vertex $x \in V_{a,b} \cap D$.
  The vertex $x$ is in $\Vactive$ as $D$ is a dominating candidate, so
  it had to survive Step \ref{step-lone-$2$-pack} --- thus it dominates $V_{a,b}$ and, in particular,
  $v$.

  If $v$ was moved to $\Vdone$ in Step \ref{step:common-leg}, we know that $v$ is in some
  pack $Z$ that shares a leg $a$ with two packs $X,Y \in \uppacks$.
  Consider $x \in X \cap D$, $y \in Y \cap D$. We know $xy \not\in E$ as $D$ is a dominating
  candidate, and $x,y \in N(a)$. Thus, by Lemma \ref{lem:two-dom-nv}, $\{x,y\}$ dominates
  $N[a]$, and, in particular, the vertex $v$.

  Now for the third claim of the Lemma, consider a \midsname $D$ compatible with $\uppacks$.
  It was a dominating candidate before we performed the steps \ref{step-useless-baloon}--\ref{step:common-leg}.
  We want to prove that it is still a dominating candidate, i.e., that no vertex of $D$ was moved
  from $\Vactive$ to $\Vdone$ by any of the steps. In the case of Step \ref{step-useless-baloon}
  this follows from Lemma \ref{lem-useless-baloon} and the branch we followed in Step \ref{step-Idisjoint}.
  In the case of Step \ref{step:common-leg} the vertices were moved to $\Vdone$ from $\Vpassive$,
  which is disjoint with $D$.

  Now assume $x \in D$ was moved to $\Vdone$ in Step \ref{step-lone-$2$-pack}. This means $x \in V_{a,b}$,
  where $V_{a,b}$ is alone in $\uppacks$, and there exists a $v \in V_{a,b}$ that is not
  dominated by $x$. As $D$ is a dominating set, however, $v$ is dominated by some $y \in D$,
  $y \neq x$. By Lemma \ref{lem:zeppelin-knows-locally} the pack $Y$ that $y$ is in (which is distinct from $V_{a,b}$ as $x,y\in D$) has to share a leg
  with $V_{a,b}$, which contradicts with the assumption that $V_{a,b}$ is alone.
\end{proof}

Using Lemma \ref{lem:end-24}, the algorithm now looks for a dominating candidate that is
a dominating set in $G$. Note that a dominating candidate is a dominating set if and only if it dominates $\Vpassive$,
since Lemma \ref{lem:end-24} ensures that the partition $(\Vactive,\Vpassive,\Vdone)$ is safe (i.e., any dominating candidate always dominates $\Vactive \cup \Vdone$).
In the following sections we reduce the sets $\Vpassive$ and $\Vactive$, sometimes by branching into a
limited number of subcases. In each branching step the subcases cover all the possibilities
for a dominating set which is a dominating candidate. Note that if at any step we decide that
a vertex $v \in \Vactive$ will not be used in the solution, we may move it directly to $\Vdone$,
as each dominating candidate dominates $v$ by the definition of a safe partition. In all steps, we shall
only move vertices to $\Vdone$ from $\Vactive$ or $\Vpassive$, not between $\Vactive$ and $\Vpassive$.
This provides us with the invariants
$\Vactive \subset \bigcup \uppacks$ and $\Vpassive \subset V \setminus (I \cup \bigcup\uppacks)$.

Let us introduce the following step.
\begin{step}\label{step:empty-active}
If at any moment, for some $X \in \uppacks$ we have $X \cap \Vactive = \emptyset$,
we terminate this branch, as there are no dominating candidates.
If at any moment, for some $v \in \Vpassive$ we have $N(v) \cap \Vactive = \emptyset$, we terminate this branch, as no dominating candidate dominates $v$.
\end{step}

If our instance has an \midsname{} of size at most \(k\), then by
the preceding arguments there exists a dominating candidate which
is also a dominating set. We now fix one such (as yet unknown)
dominating candidate which is a dominating set, and refer to it as
\emph{the solution}.

  \begin{figure}[t]
    \centering
    \includegraphics[clip,scale=0.35]{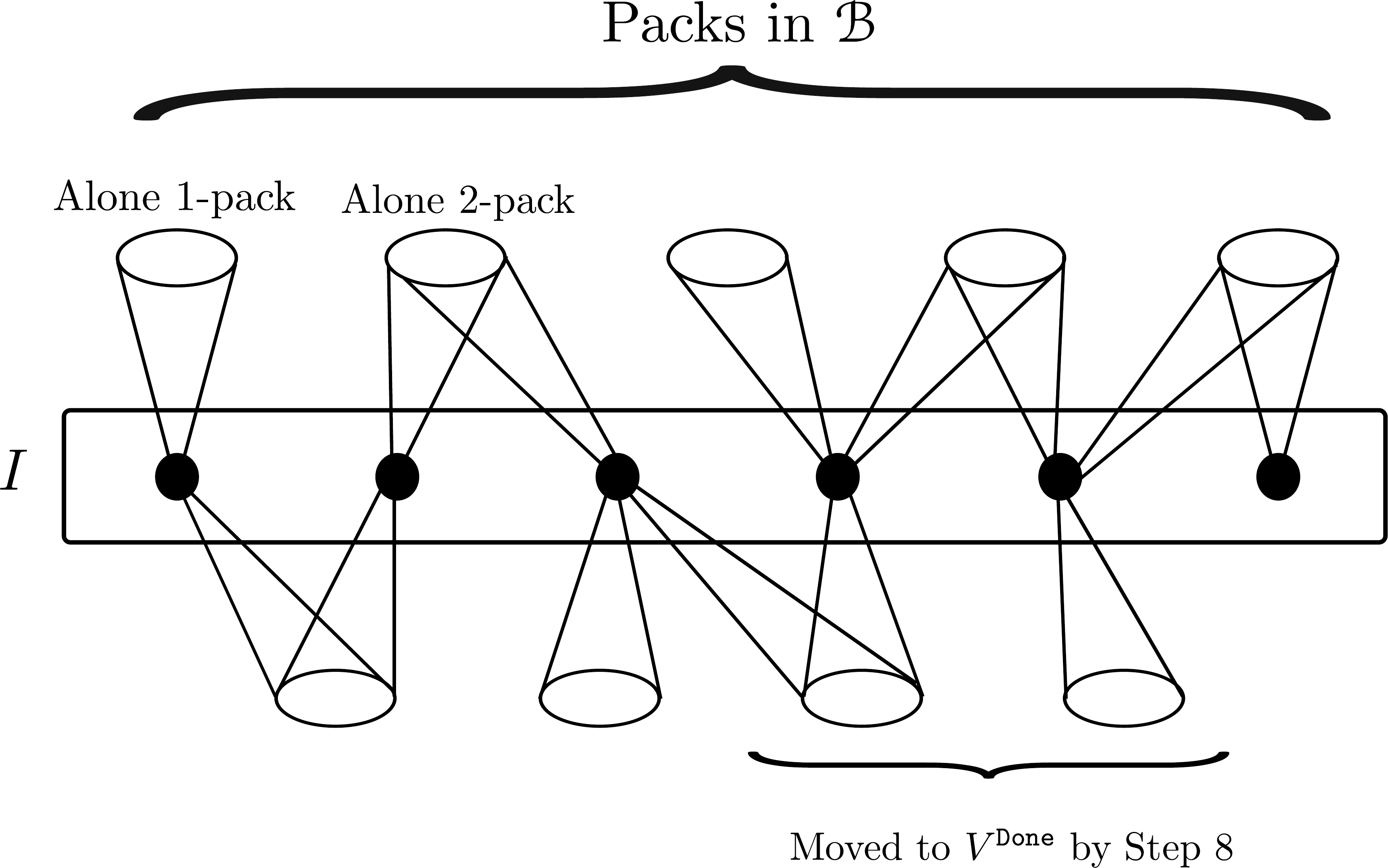}
    \caption{\label{fig:packs2}A snapshot of the graph after the
      application of Step~\ref{step:empty-active}. Edges with
      end-points in different packs are not shown.}
  \end{figure}


\newcommand{\balloons}{\ensuremath{\mathtt{Packs}_1}}
\newcommand{\Tzero}[1]{\ensuremath{\mathtt{T0}^{#1}}}
\newcommand{\Tzerop}[1]{\ensuremath{\mathtt{T0ext}^{#1}}}
\newcommand{\Tone}[2]{\ensuremath{\mathtt{T1}_{#2}^{#1}}}
\newcommand{\Ttwoall}{\ensuremath{\mathtt{T2}}}
\newcommand{\Ttwo}[1]{\ensuremath{\Ttwoall^{#1}}}
\newcommand{\clusters}{\ensuremath{\mathtt{Clusters}}}

\subsection{Decomposition of $1$-packs}\label{ss:$1$-packs}

In this section we look into the structure of $1$-packs, i.e., sets $V_a$ for $a \in I$.
Recall that each $G[V_a]$ is a clique by Lemma \ref{lem:balloon-clique}.
Let $\balloons$ be the set of all $1$-packs.
\begin{definition}
  Let $V_a$ be a $1$-pack. We partition the vertices in $V_a$ into the following sets, depending
  on their neighbourhood in $\bigcup \balloons$
  \begin{enumerate}
  \item $\Tzero{a}$ consists of those vertices $v \in V_a$ that do not know any other $1$-pack
     except for $V_a$, i.e., $N[v] \cap \bigcup \balloons \subset V_a$;
   \item $\Tone{a}{b}$ consists of those vertices $v \in V_a$ that know only $1$-packs $V_a$ and $V_b$
     for $a \neq b \in I$, i.e., $N[v] \cap \bigcup \balloons \subset V_a \cup V_b$;
   \item $\Ttwo{a}$ consists of all remaining vertices in $V_a$, i.e., those that know vertices
     from at least two $1$-packs other than $V_a$.
 \end{enumerate}
\end{definition}
Naturally, the sets $\Tzero{a}$, $\Tone{a}{b}$ or $\Ttwo{a}$ may be empty.
For example, if $V_a$ consists of a single vertex, it belongs to one of those sets
and the other two are empty.

Note that Step \ref{step:up-T0} moved $\Tzero{a}$ to $\Vdone$ for all $V_a \in \uppacks$.

Now, for each $1$-pack $V_a \in \uppacks$ we guess its part from which a vertex is taken to the solution.
\begin{step}\label{step:up-$1$-packs}
  For each $1$-pack $V_a \in \uppacks$ guess one nonempty set $T \in \{\Ttwo{a}\} \cup \{\Tone{a}{b} : b \in I, b \neq a\}$. The solution is only allowed to take a vertex from $T$, i.e.,
  we move all vertices from $V_a \setminus T$ to $\Vdone$.
\end{step}
Note that there are $O(k)$ choices for each $1$-pack, so Step \ref{step:up-$1$-packs} leads to $2^{O(k \log k)}$ subcases.

Now we switch to analyzing sets $\Ttwo{a}$.
\begin{lemma}\label{lem:Ttwo}
  Let $\Ttwoall = \bigcup_{a \in I} \Ttwo{a}$.
  Let $G_\Ttwoall$ be the graph with vertex set $\Ttwoall$ and edge set
  consisting of those edges in $G[\Ttwoall]$ that have endpoints in different $1$-packs.
  Take two vertices $v, w \in \Ttwoall$, $v \in V_a$, $w \in V_b$, $a \neq b$.
  Then $v$ and $w$ are connected by an edge in $G$ (equivalently in $G_\Ttwoall$)
  if and only if $v$ and $w$ are in the same connected component of $G_\Ttwoall$.
\end{lemma}
\begin{proof}
  The forward implication is trivial. For the other direction, assume for the sake of contradiction that \(v,w\) are in the same component of \(G_{T2}\) but $vw \notin E$. Let $v=v_0,v_1, \ldots, v_k=w$ be a fixed shortest path in $G_\Ttwoall$
  between $v$ and $w$. Let $V_{a_i}$ be the $1$-pack containing vertex $v_i$.

  Note that $a_iv_i, v_iv_{i-1}, v_iv_{i+1} \in E$ and $a_iv_{i-1},a_iv_{i+1} \notin E$ (consecutive
  vertices on the path are in different $1$-packs by the definition of $G_\Ttwoall$).
  Thus we have $v_{i-1}v_{i+1} \in E$, as otherwise we have the claw $\claw{v_i}{a_i}{v_{i-1}}{v_{i+1}}$.
  If $a_{i-1} \neq a_{i+1}$, then $v_{i-1}v_{i+1}$ would be an edge in $G_\Ttwoall$ and
  the chosen path would not be the shortest. Thus, $a_i = a_{i+2}$ for all $0 \leq i \leq k-2$,
  i.e., the path oscillates between two $1$-packs. Note that in this case $a_1 = b$.

  As $v \in \Ttwo{a}$, we have a neighbour $u$ of $v$ that is in different $1$-pack than $V_b$,
  say $u \in V_c$. We now prove by induction that $uv_i \in E$. The base of the induction is satisfied:
  $uv_0 = uv \in E$. For the induction step, assume $uv_i \in E$. Then we risk the claw $\claw{v_i}{u}{a_i}{v_{i+1}}$:
  $v_iu \in E$ (by the induction assumption), $v_ia_i \in E$, $v_iv_{i+1} \in E$, $ua_i \notin E$ as
  $c \neq a_i$ and $a_iv_{i+1} \notin E$ as $a_i \neq a_{i+1}$. Thus $uv_{i+1} \in E$.

  Therefore $\{v_0, v_1, \ldots, v_k\} \subset N[u] \setminus N[c]$, and, by Lemma \ref{lem:nv-nw-clique},
  $vw = v_0v_k \in E$.
\end{proof}
\begin{lemma}\label{lem:tone-tone}
  For any $1$-packs $V_a$, $V_b$, $a \not = b$ we have $N(\Tone{a}{b}) \cap V_b = \Tone{b}{a}$
  and $N(\Ttwo{a}) \cap V_b \subset \Ttwo{b}$.
\end{lemma}
\begin{proof}
  Let $v \in \Ttwo{a}$ and let $v_b \in V_b \cap N(v)$, $a \neq b$.
  By the definition of $\Ttwo{a}$, $v$ has got neighbours in at least
  two $1$-packs other than $V_a$, so let $v_c \in V_c \cap N(v)$, $a \neq c \neq b$.
  We risk a claw $\claw{v}{a}{v_b}{v_c}$: $va, vv_b,vv_c \in E$, $av_b \notin E$ and $av_c \notin E$.
  Thus $v_bv_c \in E$, $v_b \in \Ttwo{b}$ and $N(\Ttwo{a}) \cap V_b \subset \Ttwo{b}$.

  Now suppose there is a vertex \(u\) which belongs to both
  \(N(\Tone{a}{b})\) and \(\Ttwo{b}\). Then \(u\in V_{b}\), there
  is a vertex \(v\in\Tone{a}{b}\subseteq V_{a}\) which is a
  neighbour of \(u\), and \(u\) sees a vertex \(w\) which belongs
  to a $1$-pack \(V_{c}\) which is different from both \(V_{a}\)
  and \(V_{b}\). Thus \(uv,ub,uw\in E\). Since \(v,w\) belong to
  $1$-packs other than \(V_{b}\), neither of them sees \(b\). Since
  \(v\in\Tone{a}{b}\), it does not see \(w\) which is in a $1$-pack
  \(V_{c}\) that is different from both \(V_{a}\) and
  \(V_{b}\). Thus \(\{vb,bw,vw\}\cap E=\emptyset\), and so the
  vertex set \(\{u,v,b,w\}\) induces a claw, a
  contradiction. Hence $N(\Tone{a}{b}) \cap \Ttwo{b} = \emptyset$.

  Obviously $N(\Tone{a}{b}) \cap \Tzero{b} = \emptyset$, so
  $N(\Tone{a}{b}) \cap V_b \subset \Tone{b}{a}$.  By symmetry,
  $N(\Tone{b}{a}) \cap V_a \subset \Tone{a}{b}$. Since every
  vertex in $\Tone{b}{a}$ has a neighbour in $V_a$, we have
  $N(\Tone{a}{b}) \cap V_b = \Tone{b}{a}$.
\end{proof}
This leads us to the following definition:
\begin{definition}
  Take the graph $G_\Ttwoall$ from Lemma \ref{lem:Ttwo}.
  The vertex set of any connected component of $G_\Ttwoall$ is
  called a {\em{cluster}}.
  By $\clusters$ we denote the set of all clusters.
\end{definition}
Observe that, in general, a $1$-pack \(V_{a}\) can have nonempty
intersections with more than one cluster. Note that by Lemma
\ref{lem:Ttwo}, each cluster induces a clique in $G$. The
structure of clusters gives us good control on what can be
dominated by a vertex in a cluster.
\begin{corollary}\label{cor:cluster-dom}
  Let $v \in V_a$ be a vertex in a cluster $C$.
  Then $N[v] \setminus N[a] = C \setminus V_a$, i.e., vertex $v$ dominates
  the cluster $C$ and some neighbours of $a$.
\end{corollary}
\begin{proof}
  By Lemma \ref{lem:zeppelin-knows-locally}, $v$ can have neighbours in $2$-packs
  with leg $a$ and in other $1$-packs. By Lemma \ref{lem:Ttwo} and Lemma \ref{lem:tone-tone},
  the set of neighbours of $v$ in other $1$-packs is exactly $C \setminus V_a$ and the
  corollary follows.
\end{proof}

We now move to $1$-packs outside $\uppacks$.
Let $V_a \notin \uppacks$ and $V_a \cap \Vpassive \neq \emptyset$, i.e., 
$V_a$ was not moved to $\Vdone$ in Step \ref{step:common-leg}. Then there exists
exactly one pack in $\uppacks$ with leg $a$, and it is a $2$-pack $V_{a,b}$
(since it is not $V_a$). Note that by Lemma \ref{lem:zeppelin-knows-locally} in any dominating candidate vertices in $\Tzero{a}$
can be only dominated from the vertex in $V_{a,b}$, or else there would be a claw. For the same reason
only $V_{a,b}$ can dominate $\Tone{a}{c}$ if $V_c \notin \uppacks$ or
in Step \ref{step:up-$1$-packs} the algorithm did not guess the set $\Tone{c}{a}$ for the $1$-pack $V_{c}$. Thus the following step leaves the algorithm in a safe state.
\begin{step}\label{step:tzeroex}
  Let $V_a \notin \uppacks$ and $V_a \cap \Vpassive \neq \emptyset$.
  Let $V_{a,b}$ be the unique pack in $\uppacks$ with leg $a$.
  Let
  $$\Tzerop{a} = \Tzero{a} \cup \bigcup \{\Tone{a}{c}: \Tone{c}{a} \cap \Vactive = \emptyset\}.$$
  Move to $\Vdone$ all vertices from $V_{a,b} \cap \Vactive$ that does not
  dominate all of $\Tzerop{a}$ (we cannot use them in the solution, since, by Lemma \ref{lem:zeppelin-knows-locally},
  only a vertex from $V_{a,b}$ can dominate $\Tzerop{a}$; recall that by Lemma \ref{lem:end-24}
  any dominating candidate dominates $\Vactive$, so we do not need to move them to $\Vpassive$).
  Move $\Tzerop{a}$ to $\Vdone$ (as it is now dominated by any vertex in $V_{a,b} \cap \Vactive$).
\end{step}

Let us analyze sets $\Tone{a}{c}$ more deeply.
\begin{lemma}\label{lem:dominate-tone}
  Let $V_a \notin \uppacks$ and $V_a \cap \Vpassive \neq \emptyset$.
  Let $V_{a,b}$ be the unique pack in $\uppacks$ with leg $a$.
  Assume that $v \in V_{a,b}$, $w \in \Tone{a}{c}$, $vw \in E$ and $c \neq b$.
  Then $V_a \setminus \Tone{a}{c} \subset N[v]$.
\end{lemma}
\begin{proof}
  Let $x$ be an arbitrary vertex in $N(w) \cap V_c$
  and let $y$ be an arbitrary vertex in $V_a \setminus \Tone{a}{c}$.
  As $wx, wy, wv \in E$ (recall that $G[V_a]$ is a clique), we risk a claw $\claw{w}{v}{x}{y}$.
  Note that $vx \notin E$ due to Lemma \ref{lem:zeppelin-knows-locally}
  and $xy \notin E$, as $x \in \Tone{c}{a}$ (Lemma \ref{lem:tone-tone})
  and $y \in V_a \setminus \Tone{a}{c}$. Thus $vy \in E$.
\end{proof}
This leads us to the following step.
\begin{step}\label{step:tone-branch}
  Let $V_a \notin \uppacks$ and $V_a \cap \Vpassive \neq \emptyset$.
  Let $V_{a,b}$ be the unique pack in $\uppacks$ with leg $a$.
  Assume that $\Tone{a}{c} \cap \Vpassive$ is nonempty for at least 
  one vertex $c \in I\setminus \{a,b\}$.
  Branch into following cases:
  \begin{enumerate}
    \item There exists $c \in I \setminus \{a,b\}$
      such that the vertex in the solution from $V_{a,b}$ dominates at least one vertex from
      $\Tone{a}{c}$. Guess $c$ (there are $O(k)$ choices).
      Move all vertices $v \in V_{a,b}$ with $N[v] \cap \Tone{a}{c} = \emptyset$
      to $\Vdone$. Move all vertices in $V_a \setminus \Tone{a}{c}$ to $\Vdone$,
      as they are dominated by every vertex in $V_{a,b} \cap \Vactive$ by
      Lemma \ref{lem:dominate-tone}.
    \item The vertex in the solution from $V_{a,b}$ does not dominate anything
      from $\Tone{a}{c}$ for any $c \in I \setminus \{a,b\}$.
      Move all vertices in $V_{a,b}$ that do not satisfy this condition to $\Vdone$.
      Note that now, for each $c \in I \setminus \{a,b\}$,
      the vertices from $\Tone{a}{c}$ can be dominated only by a vertex from $V_c$,
      as no vertex from $V_a$ is in the solution.
      Thus, for each $c \in I \setminus \{a,b\}$ we move to $\Vdone$ all vertices
      in $V_c \cap \Vactive$ that do not dominate all of $\Tone{a}{c}$, and all
      vertices in $\Tone{a}{c}$, as they are now guaranteed to be dominated.
  \end{enumerate}
  Note that we move all vertices from $\Vactive$ directly to $\Vdone$ (not to $\Vpassive$)
  as they are guaranteed to be dominated by any dominating candidate by Lemma \ref{lem:end-24}.
\end{step}
For each $1$-pack $V_a$ we have $O(k)$ choices, so the number of subcases
here is $2^{O(k \log k)}$. We claim
that at this point for each $V_a \notin \uppacks$ we may have
$\Tone{a}{c} \cap \Vpassive \neq \emptyset$ for at most one choice of $c \in I \setminus \{a\}$.
Indeed, if in Step \ref{step:tone-branch} we have branched into the first case and guessed
$c \in I \setminus \{a,b\}$, only $\Tone{a}{c}$ may remain nonempty. Otherwise, only
$\Tone{a}{b}$ may remain nonempty.

We now aim to move sets $\Ttwo{a} \cap \Vpassive$ to $\Vdone$.
The following lemma shows some more of the structure of clusters.
\begin{lemma}
  Let $V_a \notin \uppacks$ and $V_a \cap \Vpassive \neq \emptyset$.
  Let $V_{a,b}$ be the unique pack in $\uppacks$ with leg $a$.
  Assume that $\Ttwo{a}$ has vertices from at least two clusters.
  Then for each vertex $v\in V_{a,b}$ either $\Ttwo{a} \subset N[v]$
  or $\Ttwo{a} \cap N[v] = \emptyset$.
\end{lemma}
\begin{proof}
  For the sake of a contradiction, assume that there exist $v \in V_{a,b}$
  and $u, w \in \Ttwo{a}$, $vu \in E$, $vw \notin E$.
  W.l.o.g. we may assume that $u$ and $w$ lie in different clusters. Indeed,
  otherwise we have a vertex $z \in \Ttwo{a}$ that lies in a different cluster
  than $u$ and $w$. If $vz \in E$, we take $u:=z$, and if $vz \notin E$, we take $w:=z$.

  Let $x$ be a neighbour of $u$ that lies in a $1$-pack different than $V_a$ and $V_b$
  (there exists one by the definition of $\Ttwo{a}$). We have a claw $\claw{u}{v}{w}{x}$:
  $uv, uw, ux \in E$ (recall that $G[V_a]$ is a clique), $vw \notin E$, $wx \notin E$
  (as $w$ and $x$ are in different clusters and in different $1$-packs) and $vx \notin E$
  (Lemma \ref{lem:zeppelin-knows-locally}), a contradiction.
\end{proof}
This suggests the following branching:
\begin{step}\label{step:ttwo1}
  Let $V_a \notin \uppacks$ and $V_a \cap \Vpassive \neq \emptyset$.
  Let $V_{a,b}$ be the unique pack in $\uppacks$ with leg $a$.
  Let $C_1, C_2, \ldots, C_d$ be clusters with vertices in $\Ttwo{a}$.
  Assume $d \geq 2$, i.e., $\Ttwo{a}$ has vertices from at least two clusters.
  We branch into two cases:
  \begin{enumerate}
  \item the vertex in the solution from $V_{a,b}$ dominates
    $\Ttwo{a}$; we move all vertices from $V_{a,b} \cap \Vactive$
    that do not dominate $\Ttwo{a}$ to $\Vdone$ and move
    $\Ttwo{a}$ to $\Vdone$.
  \item the vertex in the solution from $V_{a,b}$ does not
    dominate any vertex from $\Ttwo{a}$; we move all vertices from
    $V_{a,b} \cap \Vactive$ that dominate $\Ttwo{a}$ to $\Vdone$.
  \end{enumerate}
\end{step}
We can also similarly take care of $1$-packs that contain vertices from exactly one cluster:
\begin{step}\label{step:ttwo2}
  Let $V_a \notin \uppacks$ and $V_a \cap \Vpassive \neq \emptyset$.
  Let $V_{a,b}$ be the unique pack in $\uppacks$ with leg $a$.
  Assume $\Ttwo{a} \neq \emptyset$ and $\Ttwo{a} \subset C$ for some cluster $C$.
  Branch into two cases:
  \begin{enumerate}
    \item the vertex in the solution from $V_{a,b}$ dominates
      $\Ttwo{a}$; as before, we move all vertices from $V_{a,b} \cap \Vactive$ that do not dominate
  $\Ttwo{a}$ to $\Vdone$ and move $\Ttwo{a}$ to $\Vdone$;
    \item the vertex in the solution from $V_{a,b}$ does not dominate whole
      $\Ttwo{a}$. As before, we move all vertices from $V_{a,b} \cap \Vactive$ that dominate
      $\Ttwo{a}$ to $\Vdone$. 
  \end{enumerate}
\end{step}

Let $C_1,C_2,\ldots,C_d$ be the clusters that are not disjoint with $\Vpassive$ after performing Steps \ref{step:ttwo1} and \ref{step:ttwo2}. For each $1 \leq i \leq d$ there exists a $1$-pack $V_{a_i} \notin \uppacks$ and a $2$-pack $V_{a_i,b_i} \in \uppacks$ such that
no vertex in $V_{a_i,b_i} \cap \Vactive$ dominates whole $V_{a_i} \cap C_i \cap \Vpassive$.
Thus, by Lemma \ref{lem:zeppelin-knows-locally}, for each $i$ the solution takes at least one vertex from cluster $C_i$.
This justifies the following branching rule:

\begin{step}\label{step:ttwo3}
   If $d > k$, return NO from this branch, as clusters $C_i$ are pairwise disjoint.
   Otherwise, for each $1 \leq i \leq d$ guess a distinct $1$-pack $B_i \in \uppacks$
      where the solution contains a vertex in $C_i$; move all vertices from $B_i \setminus C_i$
      and $(C_i \setminus B_i) \cap \Vpassive$ to $\Vdone$.
      We say that the $1$-pack $B_i$ is {\em{guessed to dominate}} $C_i$.
\end{step}

Note that in the above steps we move all vertices from $\Vactive$ directly to $\Vdone$ and not to $\Vpassive$,
as they are dominated by any dominating candidate by Lemma \ref{lem:end-24}.
Note also that after performing Steps \ref{step:ttwo1}, \ref{step:ttwo2} and \ref{step:ttwo3}, we have moved all sets $\Ttwo{a}$ to $\Vdone$. 

Moreover, in Steps \ref{step:ttwo1}, \ref{step:ttwo2} for each of $O(k)$ $1$-packs we have guessed one of two possible options, and in Step \ref{step:ttwo3}, for each of at most $k$ clusters we have guessed one of $O(k)$ possible options. This leaves us with $2^{O(k\log k)}$ branches after performing Steps \ref{step:ttwo1}, \ref{step:ttwo2} and \ref{step:ttwo3}.



We now perform some cleaning.
\begin{lemma}\label{lem:clean-useless-clusters}
  Let $V_c \in \uppacks$ be a $1$-pack with $V_c \cap \Vactive \subset \Ttwo{c}$,
  i.e., the algorithm guessed in Step \ref{step:up-$1$-packs} that the vertex
         from $V_{c}$ in the solution is contained in
         $\Ttwo{c}$. Assume that $V_c$ was not guessed to dominate any cluster
  in Step \ref{step:ttwo3}. Then if there exists a dominating
  candidate $D$ that dominates $\Vpassive$, then $D' := \{c\} \cup (D \setminus V_c)$ is
  a dominating set in $G$.
\end{lemma}
\begin{proof}
  Since $D$ dominates $\Vpassive$, $D$ is a dominating set in $G$. Let $\{v\} = D \cap V_c$
  and let $C$ be the cluster containing $v$ (recall that
  $D \subset \Vactive$ and $V_c \cap \Vactive \subset \Ttwo{c}$). To prove that $D'$ is
  a dominating set in $G$ we need to ensure that $C \setminus V_c$ is dominated
  by $D \setminus \{v\}$ (recall Lemma \ref{cor:cluster-dom}).

  Take $w \in C \setminus V_c$, let $w \in \Ttwo{a}$. If $V_a \in \uppacks$,
  $w$ is dominated by a vertex from $D \cap V_a$. So let us assume that $V_a \notin \uppacks$.

  As Steps \ref{step:ttwo1}, \ref{step:ttwo2} and \ref{step:ttwo3} moved $\Ttwo{a} \cap \Vpassive$ to $\Vdone$, $w \in \Vdone$.
  We consider the possible steps in which vertex $w$ could have been placed placed in $\Vdone$.
  We moved vertices from $\Vpassive$ to $\Vdone$ in Step \ref{step:common-leg},
  Step \ref{step:tzeroex}, Step \ref{step:tone-branch}, Step \ref{step:ttwo1}, Step \ref{step:ttwo2}
  and Step \ref{step:ttwo3}.

  If $w$ was placed in $\Vdone$ in Step \ref{step:common-leg}, $D \setminus \{v\}$
  contains the two vertices in packs with leg $a$, and thus $w$ is dominated.

  Step \ref{step:tzeroex} does not touch the set $\Ttwo{a}$.

  If $w$ was placed in $\Vdone$ in Step \ref{step:tone-branch}, then the algorithm guessed that it is dominated by a vertex 
  from the $2$-pack $V_{a,b}$. As $v\notin V_{a,b}$, $D\setminus \{v\}$ dominates $w$.

  Consider Steps \ref{step:ttwo1}, \ref{step:ttwo2}, and \ref{step:ttwo3}. In the $1$-pack $V_a$,
  we either guessed to dominate whole $\Ttwo{a}$ by a vertex from the $2$-pack
  $V_{a,b} \in \uppacks$ or we guessed a $1$-pack $V_d$ ($d \neq c$) to dominate cluster $C$.
  As $v \notin V_{a,b}$ and $v \notin V_d$ respectively, $D \setminus \{v\}$ dominates $w$ in both cases.
\end{proof}
Lemma \ref{lem:clean-useless-clusters} implies that we can discard
those subcases where there exists a $1$-pack \(V_{c}\) which
satisfies the conditions of the lemma : $V_c\in\uppacks$, it was
not guessed to dominate any cluster in Step~\ref{step:ttwo3}, and
$V_c \cap \Vactive \subset \Ttwo{c}$. Indeed, if in such a subcase
there exists a solution, i.e., a dominating candidate $D$ that
dominates $\Vpassive$, by Lemma \ref{lem:clean-useless-clusters}
there exists a dominating set $D'$ not disjoint with $I$. By
Proposition \ref{prop:mds-vs-mids} ($I_D = D' \cap I$), there
exists an \midsname{} not disjoint with $I$, a contradiction to
the guess in Step \ref{step-Idisjoint}.
\begin{step}\label{step:clean-useless-clusters}
  If there exists a $1$-pack $V_c$ satisfying the conditions in Lemma \ref{lem:clean-useless-clusters},
  terminate the branch.
\end{step}

Let us conclude this section with the following lemma.
\begin{lemma}\label{lem:$1$-packs-final}
  After executing Steps \ref{step:up-$1$-packs}--\ref{step:clean-useless-clusters}:
  \begin{enumerate}
    \item the algorithm is in a safe state;
    \item if before Steps \ref{step:up-$1$-packs}--\ref{step:clean-useless-clusters} there existed
      a dominating candidate that dominated $\Vpassive$, then after Steps
      \ref{step:up-$1$-packs}--\ref{step:clean-useless-clusters}
      there exists one in at least one subcase
      or there exists an \midsname{} not disjoint with $I$;
    \item we branched into at most $2^{O(k \log k)}$ subcases;
    \item in every $1$-pack $V_a \notin \uppacks$, the set
      $V_a \cap \Vpassive$ is empty or is contained in one set $\Tone{a}{c}$;
    \item in every $1$-pack $V_a \in \uppacks$, the set
      $V_a \cap \Vactive$ is contained in one set $\Tone{a}{b}$ or
      in one cluster in $\Ttwo{a}$.
  \end{enumerate}
\end{lemma}
\begin{proof}
  The first two claims were justified by the inline comments when steps were described.

  The third claim can be seen as follows.
  In Step \ref{step:up-$1$-packs}, in Step \ref{step:tzeroex} 
  and in Step \ref{step:clean-useless-clusters} we do not branch.
  In Step \ref{step:tone-branch} we have $O(k)$ subcases for each $1$-pack
  $V_a \notin \uppacks$. As we have $O(k)$ $1$-packs, the bound holds for this step.
  The bound on the number of subcases introduced by Steps \ref{step:ttwo1}, \ref{step:ttwo2}, \ref{step:ttwo3}
  has been justified after their descriptions.

  As for the fourth claim, note that after Step \ref{step:ttwo1} and Step \ref{step:ttwo2},
  the sets $\Ttwo{a}$ are contained in $\Vdone$. Step \ref{step:tzeroex} moved sets $\Tzero{a}$
  to $\Vdone$. Step \ref{step:tone-branch} reduced the number of sets $\Tone{a}{b}$ with
  passive vertices to at most one set.

  The fifth claim follows directly from branching in Step \ref{step:up-$1$-packs}
  and from cleaning in Step \ref{step:clean-useless-clusters}.
\end{proof}

\newcommand{\cspvars}{\ensuremath{\mathtt{Vars}}}
\newcommand{\cspcons}{\ensuremath{\mathtt{Cons}}}
\newcommand{\cspvalues}[1]{\ensuremath{\mathtt{Val}(#1)}}
\newcommand{\cspallow}[1]{\ensuremath{\mathtt{Allow}_{#1}}}

\subsection{Auxiliary CSP and dynamic programming}\label{ss:dp}

We now define an auxiliary CSP problem and see that the current state of the algorithm
is in fact an instance of this CSP.
\begin{definition}
  An instance of the auxiliary CSP consists of a set $\cspvars$ of variables,
  for each variable $x \in \cspvars$ a set of possible values $\cspvalues{x}$,
  and a set of constraints $\cspcons$. A constraint is a triple $C=(x_C,y_C,\cspallow{C})$,
  where $x_C,y_C \in \cspvars$, $x_C \neq y_C$ and $\cspallow{C} \subset \cspvalues{x_C} \times \cspvalues{y_C}$.
  The solution is an assignment $\phi$ that assigns to each $x \in \cspvars$ a value $\phi(x) \in \cspvalues{x}$
  such that for each constraint $C=(x_C,y_C,\cspallow{C})$ we have $(\phi(x_C),\phi(y_C)) \in \cspallow{C}$.
\end{definition}

If an instance of the auxiliary CSP problem has a certain simple
structure, then it can be solved in polynomial time.

\begin{lemma}\label{lem:solve-csp}
  If an auxiliary CSP instance has the property that for each $x
  \in \cspvars$ there are at most $2$ other variables such that
  there exists constraints bounding $x$ and these variables, then
  the instance can be solved in polynomial time.
\end{lemma}
\begin{proof}
  Let \(\mathcal{C}\) be an auxiliary CSP instance on a set
  \(\cspvars\) of variables which has the stated property. Let
  \(\{x,y\}\subseteq \cspvars\) be a set of two variables such
  that there is more than one constraint involving \(x\) and
  \(y\), and let these constraints be
  \(\{(x,y,A_{1}),(x,y,A_{2}),\ldots,(x,y,A_{\ell})\}\). We may
  replace all these constraints by the single constraint
  \((x,y,\displaystyle{\bigcap_{i=1}^{\ell}A_{i}})\) to obtain an
  equivalent CSP instance. Also, one can merge two constraints
  \((x,y,A_{1})\) and \((y,x,A_{2})\) which differ only in the
  order of the variables, into a single constraint \((x,y,A_{12})\) in
  the natural manner. Therefore in the rest of the proof we
  assume, without loss of generality, that there is at most one
  constraint in the auxiliary CSP instance \(\mathcal{C}\) which
  involves any given subset of two variables.

  We represent \(\mathcal{C}\) as a graph \(\mathcal{G}\) on the
  vertex set \(\cspvars\) by adding, for each constraint
  \(C=(x,y,\cspallow{C})\), an edge labelled \(\cspallow{C}\)
  between the vertices \(x\) and \(y\). Observe that because of
  the special property of \(\mathcal{C}\), this graph has maximum
  degree at most \(2\), and so it is a collection of paths and
  cycles. For any vertex set \(X\subseteq V(\mathcal{G})\), we
  define the ``sub-instance'' of \(\mathcal{C}\) associated with
  \(X\) to be the CSP instance consisting of the variable set
  \(X\), the sets of possible values of the variables in \(X\),
  and all the constraints of \(\mathcal{C}\) which involve the
  variables in \(X\). Note that, in general, the sub-instance
  associated with a vertex set \(X\) may not be well-formed, in
  that it may contain constraints which involve variables which
  are \emph{not} in \(X\).

  Let \(A\subseteq V(\mathcal{G})\) be a set of vertices of
  \(\mathcal{G}\) such that the subgraph induced by \(A\) is a
  connected component of \(\mathcal{G}\). Observe that the
  connectivity of \(\mathcal{G}[A]\) ensures that for any variable
  \(x\in A\), the set
  \(\{y\in\cspvars\,\mid\,(x,y,\cspallow{C})\in\cspcons\}\) is a
  subset of \(A\). So the sub-instance associated with \(A\) is
  well-formed. Further, if \(A_{1},A_{2},\ldots,A_{\ell}\) are the
  vertex sets of all the connected components of \(\mathcal{G}\),
  and \(\phi_{1},\phi_{2},\ldots,\phi_{\ell}\) are solutions to
  the sub-instances of \(\mathcal{C}\) associated with
  \(A_{1},A_{2},\ldots,A_{\ell}\), respectively, then
  \(\phi=\phi_{1}\uplus\phi_{2}\ldots\uplus\phi_{\ell}\) is a
  solution of \(\mathcal{C}\). Conversely, if \(\phi\) is a
  solution of \(\mathcal{C}\), then for any \(1\le i\le\ell\),
  \(\phi\) restricted to the variable set \(A_{i}\) is clearly a
  solution of the sub-instance of \(\mathcal{C}\) associated with
  \(A_{i}\).

  If the connected component induced by the vertex set \(A\) is a
  path, say \((a_{1},a_{2},\ldots,a_{\ell})\), then we can find a
  solution for the sub-instance associated with \(A\), if it
  exists, by ``pruning the path''. We first associate, with each
  \(a_{i}\), a list \(L_{i}\) containing the set
  \(\cspvalues{a_{i}}\) of possible values of \(a_{i}\). For each
  \(2\le i\le\ell\) in this order, we go through the list
  \(L_{i}\) and delete all those values \(y\in L_{i}\) for which
  there is no \(x\in L_{i-1}\) such that
  \((x,y)\in\cspallow{C};(a_{i-1},a_{i},\cspallow{C})\in\cspcons\). Observe
  that after this step, for each value \(y\in L_{i}\) there is at
  least one value \(x\in L_{i-1}\) such that assigning the values
  \(x\) to \(a_{i-1}\) and \(y\) to \(a_{i}\) satisfies the
  constraint involving \(a_{i-1}\) and \(a_{i}\).
   
  If this procedure deletes all the values in any list \(L_{i}\),
  then there is no solution for the sub-instance associated with
  \(A\), and so also for the CSP instance
  \(\mathcal{C}\). Otherwise, this sub-instance has at least one
  solution. To find such a solution, pick any surviving value
  \(x_{\ell}\in L_{\ell}\). Now for each \(\ell-1\ge i\ge 1\), in
  this order, find a value \(x_{i}\in L_{i}\) such that assigning
  the values \(x_{i}\) to \(a_{i}\) and \(x_{i+1}\) to \(a_{i+1}\)
  satisfies the constraint involving \(a_{i}\) and
  \(a_{i+1}\). Such a value \(x_{i}\) always exists, and the
  assignment which gives the value \(x_{i}\) to \(a_{i}\) for each
  \(1\le i\le\ell\) satisfies all the constraints involving the
  variables of \(A\).

  If the connected component induced by the vertex set \(A\) is a
  cycle, say \((a_{1},a_{2},\ldots,a_{\ell},a_{1})\), then we
  guess a value --- say \(x\) --- for the variable \(a_{2}\) and
  check whether there is a solution for the sub-instance
  associated with \(A\) which gives the value \(x\) to
  \(a_{2}\). To do this, we delete the vertex \(a_{2}\) from \(A\)
  to obtain a path, and associate, with each remaining \(a_{i}\),
  a list \(L_{i}\) containing the set \(\cspvalues{a_{i}}\) of
  possible values of \(a_{i}\). From the list \(L_{1}\) we delete
  all those values \(y\) for which
  \((y,x)\notin\cspallow{C};(a_{1},a_{2},\cspallow{C})\in\cspcons\). Similarly,
  from the list \(L_{3}\) we delete all those values \(y\) for
  which
  \((x,y)\notin\cspallow{C};(a_{2},a_{3},\cspallow{C})\in\cspcons\). We
  now prune the path \((a_{3},a_{4},\ldots,a_{\ell},a_{1})\) in
  the same way as before, starting with these values for the lists
  \(L_{i}\).
  
  We solve for each connected component of \(\mathcal{G}\) in this
  manner. If any component does not have a solution, then we stop
  the processing and return NO as the answer. Otherwise we return
  the disjoint union of the satisfying assignments computed for
  each component.

  Since the possible set of values and the set of constraints are
  both part of the input, a straightforward implementation of the
  pruning operation takes \(O(n^{3})\) time over all component
  paths where \(n\) is the size of the input. Also, a value for a
  variable can be guessed in \(O(n)\) time, and so a simple
  implementation of the above algorithm solves the problem in
  \(O(n^{4})\) time.
  %
  %
\end{proof}

Before we start to encode the state of our algorithm, we need one more step.
\begin{step}\label{step:cleaning}
  Let $v \in \Vpassive$.
  Assume that $N(v) \cap \Vactive \subset X$ for one
  pack $X \in \uppacks$. Then $v$ can be dominated only by the single vertex
  from the solution from $X$, so move to $\Vdone$ the vertex $v$ and all vertices
  from $X \cap \Vactive$ that do not dominate $v$. Note that by Lemma \ref{lem:end-24}
  all vertices in $X \cap \Vactive$ are dominated by any dominating candidate, so we can
  move them directly to $\Vdone$ instead of $\Vpassive$.
\end{step}
Observe that after performing Step \ref{step:cleaning} exhaustively, each vertex from $\Vpassive$ has neighbours in at least two packs from $\uppacks$ (recall that by Step \ref{step:empty-active} each vertex in $\Vpassive$ has at least one neighbour in $\Vactive$). This can be streghtened to the following observation.
\begin{lemma}\label{lem:two-degree}
  Assume we have executed Step \ref{step:cleaning} exhaustively.
  Let $W$ be a pack not in $\uppacks$ and assume that $W \cap \Vpassive \neq \emptyset$.
  Then there exist two packs $Y, Z \in \uppacks; Y\neq Z$ such that every vertex $v \in W \cap \Vpassive$
  has got neighbours in $Y \cap \Vactive$, in $Z \cap \Vactive$ and no
  other active neighbours in other packs in $\uppacks$.
  Moreover, if a pack $Y' \in \uppacks$ shares a leg with $W$, then $Y' \in \{Y,Z\}$.
\end{lemma}
\begin{proof}
  As Step \ref{step:cleaning} cannot be executed more, each vertex $v \in W \cap \Vpassive$
  has active neighbours in at least two packs in $\uppacks$. Thus, we need to prove
  that the active neighbours of $v$ are contained in only two packs from $\uppacks$.

  Firstly assume that $W$ is a $1$-pack, $W=V_a$.
  As $W$ was not moved to $\Vdone$ in Step \ref{step:common-leg},
  there exists exactly one pack in $\uppacks$
  with leg $a$, denote it by $V_{a,b}$ (it is not a $1$-pack, since it is not $V_a$).
  Moreover, by Lemma \ref{lem:$1$-packs-final}, $V_a \cap \Vpassive$ is contained
  in one set $\Tone{a}{c}$. Thus, $v$ has active neighbours only in $V_{a,b}$ and $V_c$.

  Now assume that $W$ is a $2$-pack, $W=V_{a,b}$.
  As $W$ was not moved to $\Vdone$ in Step \ref{step:common-leg},
  there exists exactly one pack in $\uppacks$ with leg $a$ (say $X_a$)
  and exactly one pack in $\uppacks$ with leg $b$ (say $X_b$). Observe that $X_a\neq X_b$ as otherwise $X_a=X_b=W$.
  Moreover, by Lemma \ref{lem:zeppelin-knows-locally}, $W$ does not have
  edges to any other pack in $\uppacks$. Thus, $v$ has active neighbours
  only in $X_a$ and $X_b$.
\end{proof}
Informally, Lemma \ref{lem:two-degree} implies that every pack
not in $\uppacks$ which still contains some nontrivial vertices
(i.e., those in $\Vpassive$) implies a constraint on only two
packs in $\uppacks$.

Using Lemma \ref{lem:two-degree} we now show how to encode the state of our algorithm after all the steps from previous sections have been performed. Recall that we have $\Vactive \subset \bigcup \uppacks$
and $\Vpassive \subset V \setminus (I \cup \bigcup\uppacks)$, as we had so
in Lemma \ref{lem:end-24} and we only performed moves from $\Vactive$ or $\Vpassive$ to $\Vdone$.
\begin{definition}
  The {\em{auxiliary CSP associated with partition $(\Vactive, \Vpassive, \Vdone)$}} is constructed as follows.
\begin{enumerate}
  \item For each pack $X \in \uppacks$ we introduce variable $x^X$ with set of values $\Vactive \cap X$.
  \item For each pair of packs $X, Y \in \uppacks$ with a common leg $a$ we introduce the constraint
    $$(x^X, x^Y, \{(v,w) \in (X \cap \Vactive) \times (Y \cap \Vactive): vw \notin E\}).$$
    This constraint is called an {\em{independence constraint}}.
  \item For each pack $W \notin \uppacks$ that has nontrivial vertices, i.e., $W \cap \Vpassive \neq \emptyset$
    take the two packs $Y$ and $Z$ from Lemma \ref{lem:two-degree} and we introduce the constraint
    $$(x^Y, x^Z, \{(v,w) \in (Y \cap \Vactive) \times (Z \cap \Vactive): W \cap \Vpassive \subset N[v] \cup N[w]\}).$$
    This constraint is called a {\em{dominating constraint}}.
\end{enumerate}
\end{definition}
The following Lemma formalizes the equivalence of the constructed auxiliary CSP and the current state of the algorithm.
\begin{lemma}
  There exists a dominating candidate $D$ that is a dominating set in $G$ if and only if
  the associated auxiliary CSP has got a solution.
\end{lemma}
\begin{proof}
  Let $D$ be a dominating candidate that is a dominating set in $G$.
  For each $x^X \in \cspvars$ define $\phi(x)$ to be the unique vertex in $D \cap X$.
  Since $D$ is a dominating candidate, $\phi$ satisfies all independence constraints.
  Since $D$ is a dominating set in $G$, in particular it dominates $\Vpassive$
  and $\phi$ satisfies all dominating constraints. Thus, $\phi$ is a solution
  to the auxiliary CSP instance.

  In the other direction, let $\phi$ be a solution to the auxiliary CSP instance.
  We prove that $D = \{\phi(x): x \in \cspvars\}$ is a dominating candidate that dominates $G$.

  By the definition of the auxiliary CSP instance, $D$ contains exactly one vertex from
  each pack in $\uppacks$, thus $D$ is compatible with $\uppacks$. The independence
  constraints imply that the second property from the dominating candidate definition
  is satisfied also.

  The dominating constraints imply that $D$ dominates $\Vpassive$. As the algorithm
  is in a safe state, this implies that $D$ dominates $G$.
\end{proof}

We have constructed the above CSP, but the multigraph associated
with it can have arbitrarily large degree.  The next section is
devoted to bounding the maximum degree of the associated
multigraph in order to use Lemma \ref{lem:solve-csp}.

\subsection{CSP degree reduction}\label{ss:degred}

In this last part of the algorithm we bound the maximum
degree 
of the multigraph associated with the auxiliary CSP problem by
\(2\), so that we can 
solve it in polynomial time as explained in Lemma
\ref{lem:solve-csp}.

Before we start, we need to do some cleaning.
\begin{step}\label{step:dom-guess}
 For each pack $W$ satisfying $W \notin \uppacks$
 and $W \cap \Vpassive \neq \emptyset$ and
 for each pack $X \in \uppacks$ that satisfies
 $N(X \cap \Vactive) \cap W \cap \Vpassive \neq \emptyset$
 guess whether the vertex in $X$ from the solution dominates
 something from $W$ or it dominates nothing from $W$.
 In both cases, move the vertices from $X \cap \Vactive$
 that do not satisfy the chosen case to $\Vdone$.
 Moreover, in the second case, apply Step \ref{step:cleaning} to pack $W$,
 as then $W \cap \Vpassive$ can be dominated by only one pack in $\uppacks$
 (Lemma \ref{lem:two-degree}).
\end{step}
Note that by Lemma \ref{lem:two-degree}, for each such $W$ there exist exactly two packs $X$.
There are $O(k^2)$ packs, thus the Step \ref{step:dom-guess} leads to 
$2^{O(k^2)}$ subcases.

After the above cleaning the following holds.
\begin{lemma}\label{lem:after-dom-guess}
  Assume that Step \ref{step:dom-guess} is performed exhaustively
  and let $C$ be a dominating constraint in the associated
  auxiliary CSP instance that corresponds to a pack $W \notin
  \uppacks$, $W \cap \Vpassive \neq \emptyset$.  Let $Y$ and $Z$
  be the packs asserted by Lemma \ref{lem:two-degree}.  Then
  each vertex in $(Y\cup Z) \cap \Vactive$ has at least one
  neighbour in $W$.
\end{lemma}
\begin{proof}
  If $W \cap \Vpassive \neq \emptyset$, then both $Y$ and $Z$ guessed in Step \ref{step:dom-guess}
  to dominate something from $W$. Thus, only vertices with neighbours in $W \cap \Vpassive$
  survived in $\Vactive$ in Step \ref{step:dom-guess}.
\end{proof}

We now present the crucial structural lemma that allows us to reduce
the auxiliary CSP instance.
\begin{lemma}\label{lem:two-cliques}
  Let $a \in I$ and $X, W_1, W_2$ be three packs with leg $a$ satisfying
  $X \in \uppacks$, $W_1, W_2 \notin \uppacks$, $W_1 \cap \Vpassive \neq \emptyset$,
  $W_2 \cap \Vpassive \neq \emptyset$. Moreover, assume that the following property
  holds: for each pack $A \in \{X, W_1, W_2\}$, for each vertex $v \in A \cap (\Vactive \cup \Vpassive)$ there exists
  a vertex $n_v \in V \setminus (X \cup W_1 \cup W_2)$ such that
  $N(n_v) \cap (X \cup W_1 \cup W_2) \subset A$. Then $(X \cup W_1 \cup W_2) \cap (\Vactive \cup \Vpassive)$
  can be partitioned
  into two sets $K_1$ and $K_2$, such that $G[K_1]$ and $G[K_2]$ are cliques
  and if $v_1 \in K_1$, $v_2 \in K_2$ and $v_1$ and $v_2$ are in different packs,
  then $v_1v_2 \notin E$.
  Such sets $K_1$ and $K_2$ can be found in polynomial time.
\end{lemma}
\begin{proof}
  Let $V_H = (X \cup W_1 \cup W_2) \cap (\Vpassive \cup \Vactive)$ and
  let $H$ be a graph with vertex set $V_H$ and with edge set $E_H$
  consisting of those edges of $G[V_H]$ that have endpoints in
  different packs. We prove that the graph $H$ has at most two connected components,
  and a vertex set of each connected component of $H$ induces a clique in $G$.

  By Lemma \ref{lem:two-degree},
  every vertex in $(W_1 \cup W_2) \cap \Vpassive$ has a neighbour in $X \cap \Vactive$.
  By Lemma \ref{lem:after-dom-guess}, every vertex in $X \cap \Vactive$
  has a neighbour in $W_1 \cap \Vpassive$ and a neighbour in $W_2 \cap \Vpassive$.
  Thus, every connected component of $H$ intersects all three packs $X$, $W_1$ and $W_2$.

  Moreover, by Lemma \ref{lem:nv-nw-clique}, for each $v \in V_H$ we have
  that $G[N[v] \setminus N[n_v]]$ is a clique. Note that $N_H(v) \subset N[v] \setminus N[n_v]$. Thus we have a following observation: if a vertex $v\in V_H$ has two neighbours in the two other packs, then they are adjacent.

  We now prove the following claim. Let $C$ be a vertex set of a connected component in $H$
  and let $v \in C \cap X$ be an arbitrary vertex. Then $C \cap (W_1 \cup W_2) \subset N[v]$.
  By the contrary, assume that there exists $w \in W_1 \cap \Vpassive \cap C$, such that $vw \notin E$.
  Let $v = v_0, v_1, v_2, \ldots, v_k = w$ be the shortest path in $H$ between $v$ and $w$;
  if $vw \notin E$ then $k \geq 2$.
  If for some $i$ the vertices $v_{i-1}$, $v_i$, $v_{i+1}$ lie in three different packs
  $X$, $W_1$, $W_2$, by the previous observation they form a triangle: $v_{i-1}$ and $v_{i+1}$ are neighbours of $v_i$ and they lie in the two other packs, so $v_{i-1}v_{i+1}\in E$. Thus the path is not the shortest
  one. Therefore, the path oscillates between $X$ and $W_1$, i.e., $v_{2i} \in X$ and $v_{2i+1} \in W_1$.
  Let $u \in W_2 \cap \Vpassive$ be an arbitrary neighbour of $v$ in $H$. Then, by induction we prove that
  $v_iu \in E$ for every $i$: $v_0u = vu \in E$ and if $v_{i-1}u \in E$, then $v_i$
  and $u$ are neighbours of $v_{i-1}$ and they lie in different packs, thus $v_iu \in E$.
  Thus $wu, vu \in E$ and $u$, $v$, $w$ lie in different packs, so $vw \in E$ and the claim
  is proven.

  Now let $v_1, v_2$ be two vertices in the same connected component $C$ of $H$
  and assume $v_1$ and $v_2$ lie in the same pack. As $C$ has vertices in each pack
  $X$, $W_1$, and $W_2$, let $u$ be a common neighbour in $H$ of $v_1$ and $v_2$ that lie in a different
  pack (it exists by the previous claim). Recall than $N[u] \setminus N[n_u]$ induces a 
  clique and $v_1, v_2 \in N[u] \setminus N[n_u]$, thus $v_1v_2 \in E$. Thus $G[C]$ is a clique.

  Assume that there are three different connected components $C_1$, $C_2$, $C_3$ in $H$.
  Take $v_1 \in C_1 \cap X$, $v_2 \in C_2 \cap W_1$, $v_3 \in C_3 \cap W_2$. We have
  $av_1,av_2,av_3 \in E$ but $v_1v_2,v_2v_3,v_3v_1 \notin E$, a contradiction, as
  $\claw{a}{v_1}{v_2}{v_3}$ is a claw.

  Thus $H$ consists of one or two connected components. If one, we take $K_1 = V_H$ and
  $K_2 = \emptyset$. If two, we take $K_1$ and $K_2$ to be equal to the vertex sets of these
  components. This completes the proof.
  Note that the sets $K_1$ and $K_2$ can be computed in polynomial time,
  since they are simply the connected components of the graph $H$.
\end{proof}
Let us note that the conditions in Lemma \ref{lem:two-cliques} can be checked in polynomial
time: for each vertex $v \in (X \cup W_1 \cup W_2) \cap (\Vpassive \cup \Vactive)$
we simply check all possibilities for $n_v$.

Note that the above lemma gives us the following step.
\begin{step}\label{step:degree-template}
  For each triple of packs $X \in \uppacks$; $W_1,W_2 \notin \uppacks$ check whether
  the conditions of Lemma \ref{lem:two-cliques} are satisfied.
  If yes, compute sets $K_1$ and $K_2$ and guess whether the
  vertex in the solution from the pack $X$ is in $K_1$ or $K_2$. If the set $K_i$ is chosen,
  move vertices from $K_{3-i} \cap X$ to $\Vdone$ (not to $\Vpassive$, as
  Lemma \ref{lem:end-24} asserts that all dominating candidates dominate $\Vactive$),
  move vertices from
  $(W_1 \cup W_2) \cap K_i \cap \Vpassive$ to $\Vdone$ (they are guaranteed to be dominated by
  the vertex in $X$), and apply Step \ref{step:cleaning}
  to the vertices in $(W_1 \cup W_2) \cap K_{3-i} \cap \Vpassive$ (now they cannot be dominated
  by the vertex from $X$).
\end{step}
Let us note that the above step moves sets $W_1 \cap \Vpassive$ and $W_2 \cap \Vpassive$ to $\Vdone$.
\begin{lemma}\label{lem:after-degree-template}
  Assume Step \ref{step:degree-template} has been executed for sets $X$, $W_1$ and $W_2$.
  Then $(W_1 \cup W_2) \cap \Vpassive = \emptyset$, i.e., $W_1$ and $W_2$ no longer give raise to
  a dominating constraint in the auxiliary CSP.
\end{lemma}
\begin{proof}
  Before Step \ref{step:degree-template} is executed on $X$, $W_1$ and $W_2$, each vertex in
  $(W_1 \cup W_2) \cap \Vpassive$ was in $K_1$ or $K_2$. Assume that $K_i$ is chosen to contain
  the vertex from the solution in $X$. Then the vertices from $(W_1 \cup W_2) \cap \Vpassive \cap K_i$
  are moved to $\Vdone$, since they are dominated by any vertex in $X \cap \Vactive$.
  Moreover, the vertices from $(W_1 \cup W_2) \cap \Vpassive \cap K_{3-i}$ are moved to $\Vdone$
  in the execution of Step \ref{step:cleaning}, since now they can be dominated only by
  vertices from one particular pack in $\uppacks$.
\end{proof}
Let us now note that Step \ref{step:degree-template} cannot be executed many times.
\begin{lemma}\label{lem:num-degree-template}
  Step \ref{step:degree-template} can be executed at most $O(k^2)$ times,
  and thus all executions lead to at most $2^{O(k^2)}$ subcases.
\end{lemma}
\begin{proof}
  If Step \ref{step:degree-template} is executed on packs $X$, $W_1$ and $W_2$, then
  $W_1 \cap \Vpassive$ and $W_2 \cap \Vpassive$ become empty. Thus
  each pack not in $\uppacks$ can be touched by Step \ref{step:degree-template}
  at most once. As there are $O(k^2)$ packs, the bound follows.
\end{proof}

We finish the algorithm with the following reasoning.
\begin{lemma}\label{lem:degree-execute}
  Assume in the auxiliary CSP instance there is a variable $x^X$
  such that there are at least three other variables $Y$ bounded
  with $X$ by a constraint
  (i.e., the variable $x^X$ has at least $3$ neighbours in the multigraph associated with
  the auxiliary CSP instance).
  Then there exists packs $W_1$ and $W_2$ such that the triple $(X,W_1,W_2)$ satisfy conditions
  for Lemma \ref{lem:two-cliques}, i.e., it is eligible for the reduction
  in Step \ref{step:degree-template}.
\end{lemma}
\begin{proof}
  We consider several subcases. In the reasoning below, we often look at various packs
  $W \notin \uppacks$, such that $W \cap \Vpassive \neq \emptyset$ and
  $X \cap \Vactive$ can dominate $W$, i.e., $W$ gives a dominating constraint that
  involve $X$. By {\em{the second dominator}} for $W$ we mean the second pack $X' \in \uppacks$
  asserted by Lemma \ref{lem:two-degree}.

  \noindent \textbf{Case 1.} $X$ is a $2$-pack, $X = V_{a,b}$. Then
  $X$ can dominate only packs with leg $a$ or $b$ (Lemma \ref{lem:zeppelin-knows-locally})
  and can be connected by independence constraints to other packs with leg $a$ or $b$.
  Recall that by Step \ref{step-stupid-uppacks} there is at most one independence constraint
  per leg of $X$.

  \noindent \textbf{Case 1.1.} $X$ is connected by independence
  constraints to two other packs $X_a$ and $X_b$, where $X_a$
  has leg $a$, and $X_b$ has leg $b$.  By Step
  \ref{step:common-leg}, all packs not in $\uppacks$ with leg
  $a$ or $b$ were moved to $\Vdone$, thus these two independence
  constraints are the only constraints that involve $X$.

  \noindent \textbf{Case 1.2.} $X$ is connected by independence constraints to one pack $X_a$
  that shares leg $a$ with $X$. By Step \ref{step:common-leg}, all packs not in $\uppacks$
  with leg $a$ were moved to $\Vdone$.
  By the assumptions of the lemma, there are at least two packs $W_1$ and $W_2$ 
  that have leg $b$, are not in $\uppacks$ and $W_1 \cap \Vpassive \neq \emptyset$
  and $W_2 \cap \Vpassive \neq \emptyset$, i.e., $W_1$ and $W_2$ induce dominating constraints.
  Moreover, we can assume that the second dominators of $W_1$ and $W_2$ are different
  and different than $X_a$, as $X$ has at least three neighbours in the multigraph associated
  with the auxiliary CSP instance.

  \noindent \textbf{Case 1.2.1.} Both $W_1$ and $W_2$ are $2$-packs, $W_1 = V_{b,c}$, $W_2 = V_{b,d}$.
  Note that $a \neq c \neq d \neq a$. Thus, $X$, $W_1$ and $W_2$ satisfy conditions for Step
  \ref{step:degree-template}, where $a$ is the private neighbour of all the vertices in $X$,
  $c$ is the private neighbour for $W_1$ and $d$ for $W_2$.

  \noindent \textbf{Case 1.2.2.} $W_1$ is a $1$-pack, $W_1=V_b$ and $W_2$ is a $2$-pack, $W_2 = V_{b,d}$.
  Recall Lemma \ref{lem:$1$-packs-final}: $W_1 \cap \Vpassive \subset \Tone{b}{c}$ for
  some $1$-pack $V_c \in \uppacks$. In other words, the $1$-pack $V_c$ is the second dominator
  for $W_1$.
  As the second dominator of $W_1$ is different than $X_a$, $X_a \neq V_c$ and $c \neq a$.
  Note that there exists at most one pack $X_d \in \uppacks$ with leg $d$, as otherwise
  $W_2=V_{b,d}$ would be moved to $\Vdone$ by Step \ref{step:common-leg}.
  Moreover, $X_d$ is the second dominator for $W_2$.
  We infer that, as the second dominators for $W_1$ and $W_2$ are different,
  $X_d \neq V_c$ and $c \neq d$. Obviously $d \neq a$. Thus, by Lemma \ref{lem:zeppelin-knows-locally}, $V_c$ has no neighbours in $X$ nor $W_2$ and
  the triple $(X,W_1,W_2)$ satisfy the condition for Step
  \ref{step:degree-template}: the private neighbour for vertices in $X$ is $a$,
  for $W_2$ is $d$, and each vertex
  in $W_1 \cap \Vpassive \subset \Tone{b}{c}$ has a neighbour in $V_c$.

  \noindent \textbf{Case 1.3.} There are no independence constraints involving $X$, i.e.,
  $X$ is an alone $2$-pack in $\uppacks$. By the assumptions of the lemma,
  for at least one leg of $X$ (say $b$) we have at least two packs $W_1$ and $W_2$ sharing leg $b$
  with $W_1 \cap \Vpassive \neq \emptyset$, $W_2 \cap \Vpassive \neq \emptyset$.

  \noindent \textbf{Case 1.3.1} Both $W_1$ and $W_2$ are $2$-packs, $W_1 = V_{b,c}$, $W_2=V_{b,d}$.
  As $a, c, d$ are pairwise different, $X$, $W_1$ and $W_2$ satisfy conditions
  for Step \ref{step:degree-template} similarly as in Case 1.2.1.

  \noindent \textbf{Case 1.3.2} $W_1=V_b$ is a $1$-pack and $W_2=V_{b,d}$ is a $2$-pack.
  Similarly as in Case 1.2.2, $W_1 \cap \Vpassive \subset \Tone{b}{c}$ and $a,c,d$
  are pairwise different ($V_a \notin \uppacks$ as $X$ is alone in $\uppacks$).
  Thus $X$, $W_1$ and $W_2$ satisfy conditions for Step \ref{step:degree-template}.

  \noindent \textbf{Case 2.} $X=V_a$ is a $1$-pack. 

  \noindent \textbf{Case 2.1.} $X$ is connected by an independence
  constraint with a pack $X' = V_{a,b}$. Then, by Step
  \ref{step:common-leg}, all packs not in $\uppacks$ with leg
  $a$ were moved to $\Vdone$. Recall that by Lemma
  \ref{lem:clean-useless-clusters} the algorithm either guessed
  that the vertex in the solution from $V_a$ dominates some
  cluster, or is contained in $\Tone{a}{c}$ for some $1$-pack
  $V_c\notin \uppacks$. In the first case $V_a$ is not bounded
  by any dominating constraint. In the second case it is bounded
  by one constraint, induced by $V_c$. Thus, $X$ can be involved
  in at most two constraints.

  \noindent \textbf{Case 2.2.} $X$ is an alone $1$-pack in $\uppacks$, i.e. it does not share legs with other packs from $\uppacks$.
  Note that by Lemma \ref{lem:$1$-packs-final}, $X \cap \Vactive \subset \Tone{a}{b}$
  for some pack $V_b$ or $X \cap \Vactive \subset C$ for some cluster $C$.

  \noindent \textbf{Case 2.2.1.} $X \cap \Vactive \subset \Tone{a}{b}$.
   By the assumptions of the lemma, there exist two packs $W_1, W_2 \notin \uppacks$ with leg $a$
  that induce a dominating constraint involving $X$. Moreover, we can assume
  that the second dominator for $V_b$, $W_1$ and $W_2$ are pairwise different.
  Let $W_1 = V_{a,c}$, $W_2 = V_{a,d}$. Observe that $b\neq c\neq d\neq b$: clearly $c\neq d$ and $b$ must be different from both of them, because otherwise the second dominator of $V_b$ would be equal to the second dominator of $W_1$ or $W_2$. Therefore, by Lemma \ref{lem:zeppelin-knows-locally}, $V_b$ do not have neighbours in $W_1$ nor $W_2$. Thus $X$, $W_1$ and $W_2$ satisfy the conditions in Step \ref{step:degree-template}: for $W_1$ and $W_2$ we take $c$ and $d$ as private neighbours, and each vertex in $X \cap \Vactive$
  has a neighbour in $V_b$.

  \noindent \textbf{Case 2.2.2.} $X \cap \Vactive \subset C$ for some cluster
  $C$. Assume that $C \cap V_b \neq \emptyset$ and $C \cap V_c \neq \emptyset$
  for some $1$-packs $V_b$ and $V_c$ (recall that a cluster has vertices in at least
  three $1$-packs). Assume in contrary, that the claim does not hold. Then there are at least three
  packs $W_1, W_2, W_3 \notin \uppacks$ with leg $a$ --- no other $1$-pack gives raise to a dominating constraint involving $X$ as $X$ was guessed to dominate cluster $C$. Let $W_i = V_{a,d_i}$ for $i=1,2,3$.
  As there are at least three such packs, we can number them so that $d_1 \neq b$ and $d_2 \neq b$.
  Then $V_b$ does not have neighbours in $W_1$ and $W_2$ and $X$, $W_1$ and $W_2$
  satisfy the conditions of Step \ref{step:degree-template}: for $W_i$ we take $d_i$ as an
  universal private neighbour and each vertex in $X \cap \Vactive$ has a neighbour in cluster $C$
  in $V_b$.
\end{proof}
\begin{corollary}\label{cor:finish}
  If Step \ref{step:degree-template} cannot be performed, the
  multigraph associated with the auxiliary CSP instance has
  maximum degree at most $2$ and it can be solved in polynomial
  time as in Lemma \ref{lem:solve-csp}.
\end{corollary}
The above corollary finishes the proof of Theorem \ref{thm:ds-claw-free-fpt}.

\subsection{Summary}\label{ss:summary}

We end this section by repeating the main ideas of the algorithm.
This subsection should not be read as an introduction to the algorithm, but rather ---
as the whole algorithm is at the same time rather complex and rather technical --- as
a tool to help the reader who followed the details to grasp the large picture.

There are two crucial steps we begin with.  The first is noting
that we can look for an \midsname{} instead of a \mdsname{}
(Proposition \ref{prop:mds-vs-mids}) --- or rather, look for a
\mdsname{} but only in the branches containing a \midsname.  The
second is noticing that we can begin with the largest independent
set, and assume that our solution is disjoint from it (otherwise
we branch on the intersection --- this is Step \ref{step:1} and
Step \ref{step-Idisjoint}). 
Note
that this trick could be done with any other set with size bounded
by $f(k)$ that can be found in FPT-time, the fact that this is the
maximal independent set is not used here.

After these two steps we can introduce packs, $1$-packs and $2$-packs. We assume the
reader who read through the whole proof is familiar with the terms by now. One
important reason this is going to be useful is that our solution will contain at most
one vertex from each pack (this is Lemma \ref{lem-packs-size-one}) --- thus, we
have in some sense localized the solution --- there are few packs (few meaning
$f(k)$, independent of $n$), so we will be able to branch over the set of packs.
We use this idea immediately in steps \ref{step-guess-uppacks} and \ref{step-stupid-uppacks} to localize the solution even further.

To get a general idea of what happens next it is good to think about the auxiliary CSP
now. The idea is that for each pack containing a vertex of the solution we have 
up to $n$ ways to choose this vertex. We think of this as of choosing a valuation for
the packs (the values being the particular vertices), and we try to see what
constraints are imposed by the fact we are looking for a \midsname.

We obtain two types of constraints --- independence and domination.
The independence constraints are always binary (that is, they always tie together only
two packs). There are, however, too many of them --- note that when looking for a 
\midsname we have an independence constraint between any two $1$-packs. Here we use
a technical trick --- we relax our assumptions, and instead of looking for a \midsname
we look for a {\em dominating candidate} (see Definition 
\ref{def-dominating-candidate}), which basically means we drop the independence
constraints between $1$-packs.

One may ask here --- why do we not drop all the independence constraints, if it is so
easy? The answer is that assuming that the solution vertices from two packs that
share a leg are independent helps us in proving domination (for instance in the
justification of Step \ref{step:common-leg}), while we will be able to control the
remaining independence constraints in Lemma \ref{lem:degree-execute}.

The situation is more involved with domination constraints. As each vertex of the
graph has to be dominated, we have $n$ domination constraints. Moreover, 
{\em a priori} a vertex can be dominated from any of the packs --- thus the
constraints are not even binary to begin with. Thus, to even define the CSP graph,
we need to deal with this problem.

To deal with the domination constraints we introduce the partition of $V$ into the
sets $\Vactive$, $\Vpassive$ and $\Vdone$. Each vertex moved to $\Vdone$ means a
domination constraint removed, each vertex removed from $\Vactive$ is a possible value
of one variable removed, and --- at the same time --- the reduction of the set of 
possible dominating candidates (and thus the possibility of performing further
reductions).

The easy part are the vertices from $\uppacks$. After some preliminary steps we
were able to show (Lemma \ref{lem:end-24}) that they will be dominated by any
dominating candidate. Thus, they do not introduce any constraints (or, to look at it
in a different way, after discarding some values of the variables that can be proved
to be unnecessary, the domination constraints imposed by these vertices are trivial).

The medium-easy part are the vertices from $2$-packs. A vertex of a $2$-pack that would
introduce a constraint on more that two variables is automatically dominated --- this
is stated in Lemma \ref{lem:two-degree}, but follows from the simple observations
around Lemmata \ref{lem:zeppelin-knows-locally} and \ref{lem:two-dom-nv}, used in the
justification of Step \ref{step:common-leg}.

The difficult part are the vertices in $1$-packs that will be dominated by other
$1$-packs. Here a whole classification needs to
be developed, to check what can each $1$-pack vertex dominate, culminating in Lemma
\ref{lem:$1$-packs-final}, which strongly localizes the vertices in $1$-packs. It helps
to understand what actually made the $1$-packs so problematic. It is mainly that while
we can pretty well control what vertices can dominate a vertex from a
$2$-pack (they have to come from a pack that shares a leg with the $2$-pack, and
after Step \ref{step:common-leg} only two of them are left), the $1$-packs can actually
be all connected to one another, and as each has only one leg, it is more difficult to
find claws in them. And the structure is indeed more complicated than in the case of
$2$-packs.

It turns out, however, that if a $1$-pack has edges into at least two other
$1$-packs, we have enough information to form claws easily, and force a strong structure
(this is the $\Ttwo{}$ case, Lemma \ref{lem:Ttwo}) --- the clusters. We analyze the
clusters to show that they do not dominate each other (Corollary \ref{cor:cluster-dom}),
and thus, in particular, there cannot be more than $k$ of them, so we will be able to
branch upon which pack dominates each cluster (Step \ref{step:ttwo1}).
On the other hand if there is only one $1$-pack
adjacent to the given one, we can branch over all possible cases (Step 
\ref{step:up-$1$-packs}).

After reducing all the constraints to be binary we are almost done.

Now we bound the degree of each vertex by $2$
, which turns out to be rather simple, although
somewhat tedious. Instead of repeating similar arguments over and
over again, we show a general framework (in Lemma
\ref{lem:two-cliques} and Step \ref{step:degree-template}), and
then apply it multiple times in Lemma \ref{lem:degree-execute}.

\section{Hardness in $t$-claw-free graphs}\label{s:hardness}
In this section we prove Theorem \ref{thm:ds-t-claw-free-whard},
  i.e., we show that the \ds{} problem is \wth{} on graph
classes characterized by the exclusion of the \(t\)-claw as an
induced subgraph, for any \(t\ge 4\). This implies that the
problem is unlikely to have FPT algorithms on these classes of
graphs~\cite{fvs2}. We prove the hardness result for the class of
\(4\)-claw-free graphs; note that this implies the result for all
\(t\ge 4\). To prove that \ds{} is \wth{} on this class, we
present a parameterized reduction from the \rbds{} problem, which
is known to be \wth{}~\cite{downey-fellows}. A direct reduction eluded us,
however, and so we make use of an intermediate, coloured version
of the problem:

\parnamedefn{\crbds{}}{A bipartite graph \(G=(R\uplus B, E)\),
  \(k\in\mathbb{N}\), and a colouring function\\
  \(c:R\to\{1,2,\ldots,k\}\)}{\(k\)}{Does there exist a set
  \(D\subseteq R\) of \(k\) distinctly coloured vertices such that
  \(D\) is a dominating set of \(B\)?}

We call such a dominating set \(D\) a \emph{colourful} red-blue
dominating set of \(G\). This coloured version turns out to be at
least as hard as the original problem:

\begin{lemma}\label{lem:crbds-hard}
  The \crbds{} problem is \wth{}.
\end{lemma}
\begin{proof}
  We reduce from the \rbds{} problem which is known to be
  \wth{}~\cite{downey-fellows}, and which is defined as follows:

  \parnamedefn{\rbds{}}{A bipartite graph \(G=(R\uplus B, E)\),
    \(k\in\mathbb{N}\)}{\(k\)}{Does there exist a set \(D\subseteq
    R\) of size \(k\) such that \(D\) is a dominating set of
    \(B\)?}

  Such a set \(D\) is called a \emph{red-blue dominating set} of
  \(G\). Observe that the above problem is equivalent to asking if
  there is a red-blue dominating set of size \emph{at most} \(k\),
  which is how this problem is usually phrased. If \(|R|<k\), then
  the problem instance is easily solved (say \YES{} if and only if
  there are no isolated vertices in \(B\)), so we can assume
  without loss of generality that \(|R|\ge k\). If there is a
  red-blue dominating set of size at most \(k\), we can always pad
  it up with enough vertices to obtain a red-blue dominating set
  of size exactly \(k\), and the converse is trivial.

  Given an instance \( (G=(R\uplus B, E),k ) \) of \rbds{}, we
  create a new graph \(G'\) whose vertex set consists of the set
  \(B\) and \(k\) copies \(R_{1},R_{2},\ldots,R_{k}\) of the set
  \(R\). For each vertex \(v\in R\), we make the neighbourhood of
  each copy of \(v\) in \(G'\) identical to the neighbourhood of
  \(v\) in \(G\); the edge set \(E'\) of \(G'\) can be thought of
  as \(k\) disjoint copies of the edge set of \(G\). We set
  \(R'=R_{1}\cup R_{2}\cup\cdots\cup R_{k}\). For each \(1\le i\le
  k\), the colouring function \(c\) maps all vertices in \(R_{i}\)
  to the colour \(i\). This completes the construction; the reduced
  instance is \((G'=(R'\cup B, E'),k,c)\). See
  Figure~\ref{fig:crbds-hardness}.
  
  \begin{figure}[t]
    \centering
    \includegraphics[clip,scale=0.4]{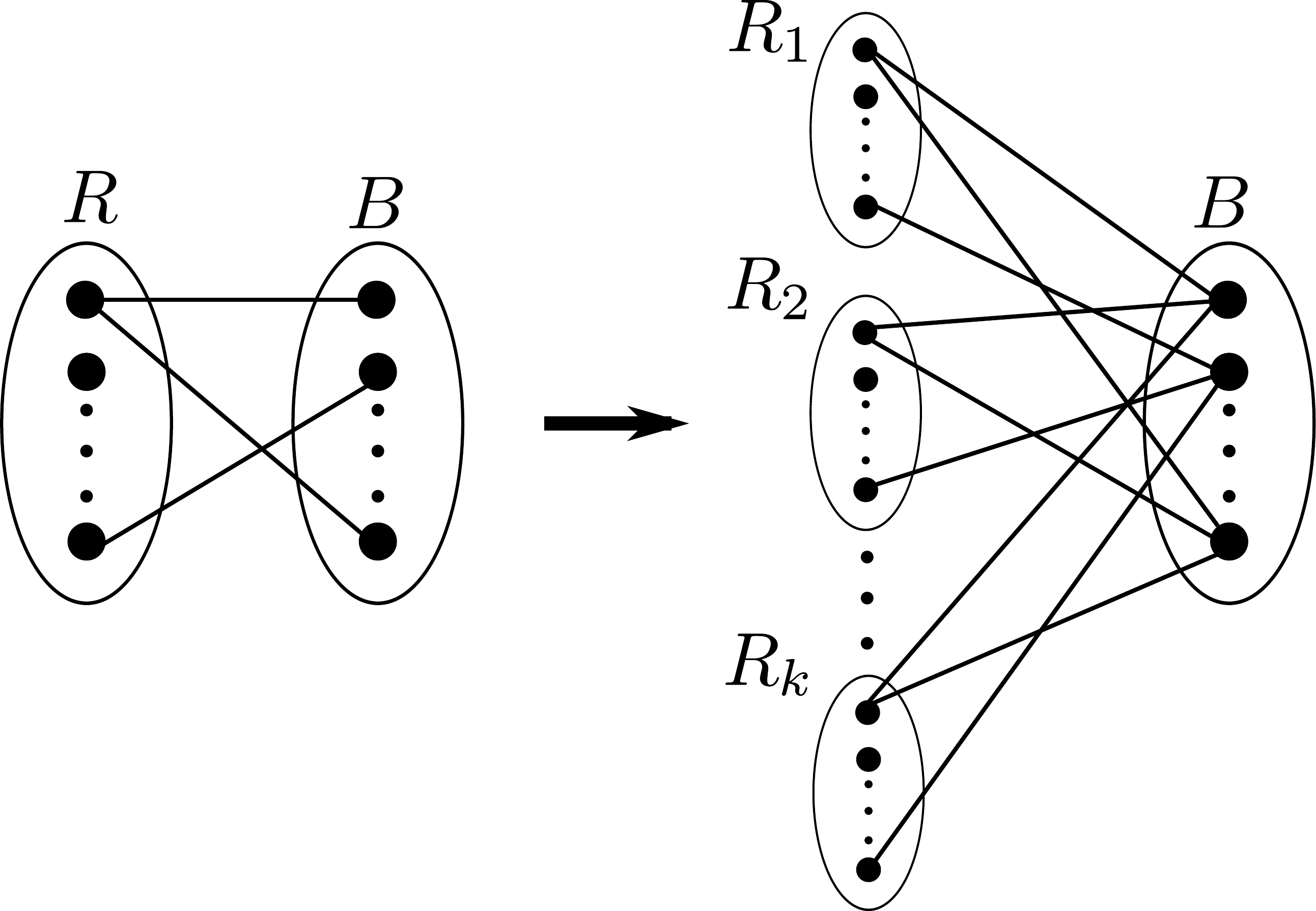}
    \caption{\label{fig:crbds-hardness}Reduction from \rbds{} to
      \crbds{}. Each set \(R_{i}\) is a copy of \(R\), and its
      vertices have a distinct colour.}
  \end{figure}

  If \((G,k)\) is a \YES{} instance of \rbds{}, then let
  \(D=\{v_{1},v_{2},\ldots,v_{k}\}\subseteq R\) be a dominating
  set of \(B\) of size \(k\). For \(1\le i,j\le k\), let
  \(v_{i}^{j}\) denote the copy of \(v_{i}\) in the set \(R_{j}\)
  in \(G'\). It is not difficult to verify that the set
  \(\{v_{i}^{i}\mid 1\le i\le k\}\) is a colourful red-blue
  dominating set of \(G'\) of size \(k\).

  Conversely, let \((G',k)\) be a \YES{} instance of
  \crbds{}. Then there exists a set of vertices
  \(D=\{v_{1},v_{2},\ldots,v_{k}; v_{i}\in R_{i}\}\) which
  dominates all vertices in \(B\), in \(G'\). Let \(D'=\{v\in
  R\mid D\text{ contains a copy of }v\}\). Then \(D'\) contains at
  most \(k\) vertices, and it is straightforward to verify that
  \(D'\) dominates \(B\) in \(G\). 
\end{proof}
We are now ready to show the main result of this section:
\begin{lemma}\label{lem:c4-free-ds-hard}
  The \ds{} problem restricted to \(4\)-claw-free graphs is
  \wth{}.
\end{lemma}
\begin{proof}
  We reduce from the \crbds{} problem, which we show to be \wth{}
  in Lemma~\ref{lem:crbds-hard}. Given an instance \((G=(R\uplus
  B,E),k,c)\) of \crbds{}, we construct an instance of \ds{} on
  \(4\)-claw-free graphs as follows. We add all possible edges
  among the vertices in \(B\) so that \(B\) induces a clique. In
  the same way, we make \(R\) a clique, and for each colour class
  (set of vertices for which \(c\) assigns the same colour)
  \(R_{i};1\le i\le k\) of \(R\), we add a new vertex \(v_{i}\)
  and make \(v_{i}\) adjacent to all the vertices in \(R_{i}\). We
  remove all colours from the vertices, and this completes the
  construction. See Figure~\ref{fig:ds_4claw-free_hardness}.

  \begin{figure}[t]
    \centering
    \includegraphics[clip,scale=0.4]{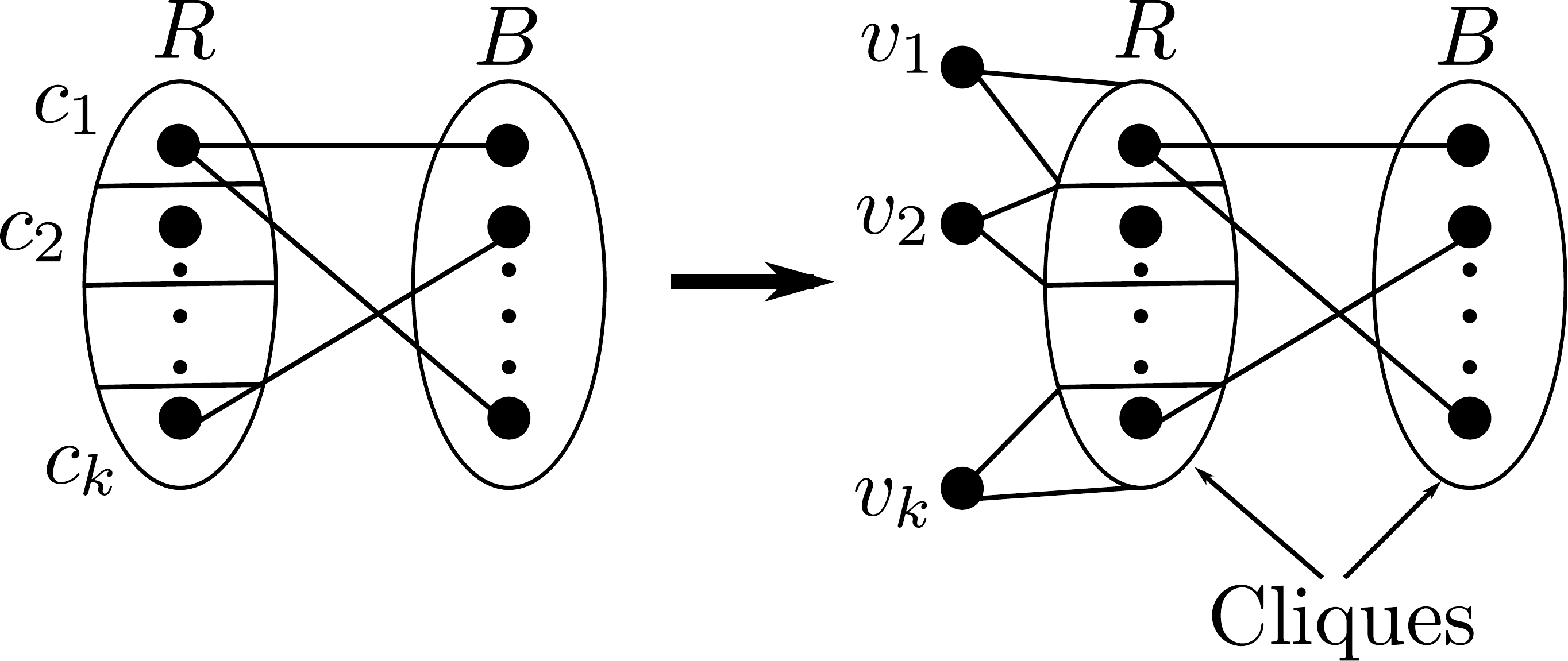}
    \caption{\label{fig:ds_4claw-free_hardness}Reduction from \crbds{} to
      \ds{} on $4$-claw-free graphs. The sets \(R,B\) are both
      made cliques, and a new vertex is made global to each colour
      class.}
  \end{figure}

  Let \(G'\) be the graph obtained. It is easy to verify that the
  neighbourhood of each vertex in \(G'\) is a union of at most
  three vertex-disjoint cliques, and so \(G'\) is a
  \(4\)-claw-free graph; \((G',k)\) is the reduced instance of
  \ds{} on \(4\)-claw-free graphs.

  If \((G,k,c)\) is a \YES{} instance of \crbds{}, then let
  \(D=\{u_{1},u_{2},\ldots,\newline u_{k}; u_{i}\in R_{i}\}\) be a
  colourful dominating set of \(B\) of size \(k\). Since we did
  not delete any edge in constructing \(G'\) from \(G\), the set
  \(D\) dominates all of \(B\) in \(G'\). Since we made the set
  \(R\) a clique in \(G'\), the set \(D\) dominates all of \(R\)
  in \(G'\). Since each new vertex that we added to \(G\) is
  adjacent to every vertex in some colour class, the set \(D\)
  dominates all the newly added vertices in \(G'\) as well. Thus
  \(D\) is a dominating set of \(G'\), of size \(k\).

  Conversely, if \((G',k)\) is a \YES{} instance of \ds{}, then
  let \(D=\{u_{1},u_{2},\ldots,u_{k}\}\) be a dominating set of
  \(G'\) of size \(k\) in \(G'\). Since the neighbourhood in
  \(G'\) of each new vertex \(v_{i}\) is the set \(R_{i}\),
  \(D\cap(R_{i}\cup\{v_{i}\})\ne\emptyset\). Since the sets
  \(R_{i}\cup\{v_{i}\};1\le i\le k\) are pairwise vertex-disjoint,
  \(D\) contains exactly one vertex from each set
  \(R_{i}\cup\{v_{i}\}\), and no other vertex. Suppose
  \(D\cap(R_{i}\cup\{v_{i}\})=v_{i}\) for some \(i\). Then we can
  replace \(v_{i}\) with an arbitrary vertex \(x\in R_{i}\), in
  \(D\), and this \(D\) would still be a dominating set of
  \(G'\). This is because the neighbourhood \(R_{i}\) of \(v_{i}\)
  is a clique, and so \(x\in R_{i}\) dominates all of
  \(R_{i}\). Thus we can assume without loss of generality that
  \(D\) contains no vertex \(v_{i}\). Thus \(D\subseteq R\) is a
  set of \(k\) vertices, one from each set \(R_{i}\), that
  dominates all vertices in \(G'\). Since we did not modify any
  adjacency between the sets \(R\) and \(B\) to construct \(G'\)
  from \(G\), it follows that in \(G\) the set \(D\) dominates all
  vertices in \(B\). Hence \(D\) is a colourful red-blue
  dominating set of \(G\) of size k.
\end{proof}

In the \name{Connected Dominating Set} (resp. \name{Dominating
  Clique}) problem, the input consists of a graph \(G\) and
\(k\in\mathbb{N}\), the parameter is \(k\), and the question is
whether \(G\) has a dominating set \(D\) of size at most \(k\)
such that the subgraph of \(G\) induced by the set \(D\) is
connected (resp. a clique). Observe that the reduction in
Lemma~\ref{lem:c4-free-ds-hard} ensures that if the reduced graph
\(G'\) has a dominating set of size at most \(k\), then it has a
dominating set \(D'\) of size at most (in fact, exactly) \(k\)
which induces a clique in \(G'\). Thus the above reduction also
shows that

\begin{corollary}\label{cor:cds-domclique-4claw-free-hard}
  The \name{Connected Dominating Set} problem and the
  \name{Dominating Clique} problem are \wth{} when restricted to
  \(4\)-claw-free graphs.
\end{corollary}

\begin{remark}\label{rem:all-claw-free}
  Observe that if a graph \(G\) contains a \(t'\)-claw \(T'\) for
  any \(t'\in\mathbb{N}\), \(G\) also contains a \(t\)-claw \(T\)
  for each \(t\le t';t\in\mathbb{N}\). Indeed, each such \(T\)
  occurs in \(G\) as an induced subgraph of \(T'\). Taking the
  contrapositive, a \(t\)-claw-free graph is also \(t'\)-claw-free
  for all \(t'\ge t;t,t'\in\mathbb{N}\). It follows that the
  hardness results stated in Lemma~\ref{lem:c4-free-ds-hard} and
  Corollary~\ref{cor:cds-domclique-4claw-free-hard} extend to
  \(t\)-claw-free graphs for all \(t\ge 4\).
\end{remark}

\section{The \clique{} problem in claw-free
  graphs}\label{s:clique}
In this section we prove Theorem \ref{thm:clique-fpt},
i.e., we give an FPT algorithm for the \clique{}
problem in $t$-claw-free graphs.
  
  The (decision version of the) Maximum Clique problem takes as
input a graph \(G\) and a positive integer \(k\), and asks whether
\(G\) contains a clique (complete graph) on at least \(k\)
vertices as a subgraph. This is one of Karp's original list of 21
NP-complete problems~\cite{Karp1972}, and the standard
parameterized version \clique{}, defined below, is a fundamental
\(W[1]\)-complete problem~\cite{downey-fellows}. The \(W[1]\)-hardness of
\clique{} implies that the problem is unlikely to have FPT
algorithms~\cite{fvs2}.

The classical decision variant of this problem remains NP-hard on
claw-free graphs~\cite[Theorem 5.4]{FaudreeFlandrinRyjacek1997}.
In this section we show that, in contrast, the problem becomes
easier from the point of view of parameterized complexity when we
restrict the input to claw-free graphs.

\begin{lemma}\label{lem:clique-fpt}
  For any \(t\in\mathbb{N}\), the \clique{} problem is FPT on
  \(t\)-claw-free graphs, and can be solved in
  \( (k+t-2)^{(t-1)(k-1)} n^{O(1)} \) time.
\end{lemma}
\begin{proof}
  We use Ramsey's theorem for graphs, which states that for any
  two positive integers \(i,c\), there exists a positive integer
  \(\mathcal{R}(i,c)\) such that any graph on at least
  \(\mathcal{R}(i,c)\) vertices contains either an independent set
  on \(i\) vertices or a clique on \(c\) vertices (or both) as an
  induced subgraph. Further, it is known~\cite{Jukna2001} that
  \(\mathcal{R}(i,c)\le {{i+c-2}\choose {c-1}}\). Setting
  \(i=t,c=k\), it follows that if a graph on at least
  \({{k+t-2}\choose{k-1}}={{k+t-2}\choose{t-1}}\le(k+t-2)^{t-1}\)
  vertices does not contain an independent set of size \(t\), then
  it must contain a clique on \(k\) vertices.

  Let \(G\) be a \(t\)-claw-free input graph for the \clique{}
  problem, and let \(v\) be any vertex in \(G\). Since \(G\) is
  \(t\)-claw-free, the neighbourhood of \(v\) contains no
  independent set of size \(t\). If \(v\) has degree at least
  \((k+t-2)^{t-1}\), it then follows from Ramsey's theorem that
  the neighbourhood of \(v\) contains a clique on \(k\)
  vertices. Hence, if any vertex in \(G\) has degree
  \((k+t-2)^{t-1}\) or more, our FPT algorithm returns \YES{};
  this check can clearly be done in polynomial time.

  Assume therefore that every vertex in the input graph has degree
  less than \((k+t-2)^{t-1}\). Our algorithm iterates over each
  vertex \(v\) of degree at least \(k-1\), and checks if its
  neighbourhood \(N(v)\) contains a clique of size
  \(k-1\). Observe that this procedure will find a \(k\)-clique in
  \(G\) if it exists.

  To check if \(N(v)\) contains a clique of size \(k-1\), the
  algorithm enumerates all \((k-1)\)-sized subsets of \(N(v)\) and
  checks whether any of these subsets induces a complete subgraph
  in \(G\). There are \({\vert N(v)\vert \choose
    {k-1}}\le{{(k+t-2)^{t-1}}\choose{k-1}}\le
  (k+t-2)^{(t-1)(k-1)}\) such subsets, and these can be enumerated
  in \(O((k+t-2)^{(t-1)(k-1)})\) time~\cite{Ehrlich1973}. For each
  subset, it is sufficient to check if all \({{k-1}\choose 2}\le
  k^{2}\) possible edges are present, which, given an adjacency
  matrix for \(G\), can be done in \(O(k^{2})\) time. Putting all
  these together, our algorithm solves the problem in
  \((k+t-2)^{(t-1)(k-1)} n^{O(1)} \) time.
\end{proof}

\section{Conclusions}\label{s:conclusions}
We derive an FPT algorithm for the \ds{} problem parameterized by
solution size, on graphs that exclude the claw \(K_{1,3}\) as an
induced subgraph. Our algorithm starts off using a maximum
independent set of the input graph, known to be computable in
polynomial time~\cite{minty:indset,sbihi:indset}. We show that it
is sufficient to look for an \emph{independent} dominating set of
the prescribed size. Our algorithm then uses the claw-freedom of
the input graph to implement reduction rules which narrow down the
possible ways in which a small dominating set could be present in
the graph. Once these rules have been exhaustively applied, we are
left with a graph and a set of constraints which must be satisfied
by every dominating set of small size, where the constraints are
highly structured in that they define an underlying graph of small
degree. We then use dynamic programming on this underlying
graph to retrieve the dominating set (or to find that no such
dominating set could exist). The algorithm uses $2^{O(k^2)}
n^{O(1)}$ time and polynomial space to check if a claw-free graph
on \(n\) vertices has a dominating set of size at most \(k\).

The most general class of graphs for which an FPT algorithm was
previously known for this parameterization of \ds{} is the class
of \(K_{i,j}\)-free graphs, which exclude, for some fixed
\(i,j\in\mathbb{N}\), the complete bipartite graph \(K_{i,j}\) as
a (not necessarily induced) subgraph~\cite{domset-philip}. To the
best of our knowledge, \emph{every} other class for which an FPT
algorithm was previously known for this parameterization of \ds{}
can be expressed as a subset of \(K_{i,j}\)-free graphs for
suitably chosen values of \(i\) and \(j\). If \(i=1\), then
\(K_{i,j}\)-free graphs are graphs of bounded degree, on which the
\ds{} problem is easily seen to be FPT. For the interesting case
when \(i,j\ge 2\), the class of claw-free graphs and any class of
\(K_{i,j}\)-free graphs are not comparable with respect to set
inclusion: a \(K_{i,j}\)-free graph can contain a claw, and a
claw-free graph can contain a \(K_{i,j}\) as a subgraph. In this
paper, we thus break new ground: we \emph{extend} the range of
graphs over which this parameterization of \ds{} is known
to be fixed-parameter tractable, beyond graph classes which can be
described as \(K_{i,j}\)-free.

In addition to this main result, we also show that the \ds{}
problem is \wth{} (and therefore unlikely to have FPT algorithms)
in \(t\)-claw-free graphs for any \(t\ge 4\), and that the
\clique{} problem is FPT in \(t\)-claw-free graphs for any
\(t\in\mathbb{N}\).

In the version of this paper which we submitted to
ArXiv~\cite{DBLP:journals/corr/abs-1011-6239}, we had stated:

``These results open up many new challenges. The most immediate
open question is to get a faster FPT algorithm with a more
reasonable running time; ideally, an algorithm that runs in
\(O^{\star}(c^{k})\) time for some small constant \(c\). Another
open problem, and perhaps of greater significance, is to find a
polynomial kernel for the problem in claw-free graphs, or to show
that no such kernel is likely to exist.''

Both these problems were later solved by Hermelin et
al.~\cite{HermelinMnichLeeuwenWoeginger2011}. Building on the
structural characterization for claw-free graphs developed
recently by Chudnovsky and
Seymour~\cite{chudnovsky:clawfree1,chudnovsky:clawfree2,chudnovsky:clawfree3,chudnovsky:clawfree4,chudnovsky:clawfree5,chudnovsky:clawfree6},
they derive an FPT algorithm for the \kDS problem on claw-free
graphs which runs in \(9^{k}n^{O(1)}\) time. They also show that
the problem has a polynomial kernel on \(O(k^{4})\) vertices on
claw-free graphs.

As mentioned above, \(K_{i,j}\)-free and claw-free graphs are two
largest classes for which we now have FPT algorithms for
\ds{}. For what other classes of graphs, not contained in these
two classes, is the problem FPT?  Finally, is there an even larger
class, which subsumes \emph{both} claw-free and \(K_{i,j}\)-free
graphs, for which the problem is FPT?

\paragraph{Acknowledgements.} We would like to thank anonymous referees for their valuable comments.

\bibliographystyle{plain}
\bibliography{domset-clawfree}

\end{document}